\documentclass[journal]{IEEEtran}
\usepackage{graphicx}
\usepackage{amsmath}
\usepackage[noend]{algpseudocode}
\usepackage{algorithmicx,algorithm}
\usepackage{amsthm}
\usepackage{amsfonts}
\usepackage{subfigure} 
\usepackage{color}
\usepackage{cite}
\usepackage{setspace}
\newtheorem{Lemma}{Lemma}
\newtheorem{Theorem}{Theorem}
\usepackage{amssymb}
\usepackage{color}
\graphicspath{{images/}}
\usepackage{booktabs}
\usepackage{multirow} 
\usepackage{makecell}
 \usepackage[numbers,sort&compress]{natbib}
 \newtheorem{remark}{Remark}
\usepackage{stfloats}
\usepackage{amsmath}
\usepackage{amssymb}
\usepackage{enumitem}
\usepackage{threeparttable}
\usepackage{comment}
\usepackage{bm}
\begin{document}
	
\title{\huge Modeling and Performance Analysis of IoT-over-LEO
	Satellite Systems under Realistic Operational Constraints: A Stochastic Geometry
	Approach}
	
\author{Wen-Yu Dong, Shaoshi Yang,~\IEEEmembership{Senior Member,~IEEE},\\ Ping Zhang,~\IEEEmembership{Fellow,~IEEE}, Sheng Chen,~\IEEEmembership{Life Fellow,~IEEE}	%
\thanks{Copyright (c) 20xx IEEE. Personal use of this material is permitted. However, permission to use this material for any other purposes must be obtained from the IEEE by sending a request to pubs-permissions@ieee.org.}
\thanks{This work was supported in part by Beijing Municipal Natural Science Foundation under Grant L242013;  and in part by the Open Project Program of the Key Laboratory of Mathematics and Information Networks, Ministry of Education, China, under Grant KF202301. \textit{(Corresponding author: Shaoshi Yang.)}}	
\thanks{Wen-Yu Dong and Shaoshi Yang are with the School of Information and Communication Engineering, Beijing University of Posts and Telecommunications, Beijing 100876, China, also with the Key Laboratory of Universal Wireless Communications, Ministry of Education, Beijing 100876, China, and also with the Key Laboratory of Mathematics and Information Networks, Ministry of Education, Beijing 100876, China (e-mail: wenyu.dong@bupt.edu.cn, shaoshi.yang@bupt.edu.cn).} %
\thanks{Ping Zhang is with the School of Information and Communication Engineering, Beijing University of Posts and Telecommunications, Beijing 100876, China, and also with the State Key Laboratory of Networking and Switching Technology, Beijing 100876, China (e-mail: pzhang@bupt.edu.cn).} %
\thanks{Sheng Chen is with the School of Electronics and Computer Science, University of Southampton, SO17 1BJ Southampton, U.K. (e-mail: sqc@ecs.soton.ac.uk).} %
\vspace*{-6mm}
}

\markboth{Accepted to appear on IEEE Internet of Things Journal, May 2025}%
{Shell \MakeLowercase{\textit{et al.}}: Bare Demo of IEEEtran.cls for IEEE Journals}

\maketitle

\begin{abstract}
The growing demand for reliable and extensive connectivity has made low Earth orbit (LEO) satellites aided Internet of Things (IoT) systems a critical area of research. However, current theoretical studies on IoT-over-LEO satellite systems often rely on unrealistic assumptions, such as infinite terrestrial areas and omnidirectional satellite coverage, leaving significant gaps in theoretical analysis for more realistic operational constraints. These constraints involve finite terrestrial area, limited satellite coverage, Earth curvature effect, integral uplink and downlink analysis, and link-dependent interference.
To address these gaps, this paper proposes a novel stochastic geometry based model to rigorously analyze the performance of IoT-over-LEO satellite systems. By adopting a binomial point process (BPP) instead of the conventional Poisson point process (PPP), our model accurately characterizes the geographical distribution of a fixed number of IoT devices in a finite terrestrial region. This modeling framework enables the derivation of distance distribution functions for both the links from the terrestrial IoT devices to the satellites (T-S) and from the satellites to the Earth station (S-ES), while also accounting for limited satellite coverage and Earth curvature effects. To realistically represent channel conditions, the Nakagami fading model is employed for the T-S links to characterize diverse small-scale fading environments, while the shadowed-Rician fading model is used for the S-ES links to capture the combined effects of shadowing and dominant line-of-sight paths. Furthermore, the analysis incorporates uplink and downlink interference, ensuring a comprehensive evaluation of system performance. The accuracy and effectiveness of our theoretical framework are validated through extensive Monte Carlo simulations. These results provide insights into key performance metrics, such as coverage probability and average ergodic rate, for both individual links and the overall system. Our study also offers an important analytical tool for optimizing the design and performance of IoT-over-LEO satellite systems with the operational constraints that are more realistic.
\end{abstract}
	

\section{Introduction}\label{S1}

\IEEEPARstart{I}{nternet} of Things (IoT) forms a crucial component of future digital infrastructure and has been attracting significant interests from both academia and industrial sectors, because of its vast potential to enhance the quality of services across various aspects of contemporary life \cite{7123563}. However, extending the current IoT that is typically based on limited terrestrial communication infrastructure to cover the vast rural and remote areas is encountering numerous problems, particularly in poor countries \cite{9042251}. In order to realize ubiquitous connectivity, IoT systems based on a large-scale low Earth orbit (LEO) satellite constellation have been considered as an appealing solution due to their convenient deployment, significant adaptability, and extensive coverage \cite{add0,9442378,9040264,9049651}. Conducting performance analysis of these systems is cost-effective and crucial for optimizing their design and ensuring their reliability in diverse environments. 

However, when conducting such performance analysis for real-world scenarios, we must consider a number of factors, such as the finite size of the IoT device distribution area, the curvature of the Earth, the limited satellite coverage, and complex propagation environments, all of which pose significant challenges to performance analysis. In particular, in remote areas lacking terrestrial infrastructure, IoT devices are often deployed in geographically constrained regions, such as mountains or deserts, where the terrain results in a limited device distribution area. This makes the widely used tractable Poisson point process (PPP) inapplicable, posing significant challenges in deriving the distribution of distances from ground IoT devices to satellites. Additionally, the curvature of the Earth and the limited satellite coverage also complicate the distance distributions and the feasibility of line-of-sight (LOS) conditions between satellites and ground nodes, making traditional flat Earth-surface based models inadequate for accurately describing these effects.  Recent research has shown that stochastic geometry \cite{1995Stochastic} and random geometric graphs \cite{2002Random}  have emerged as promising solutions for characterizing and evaluating the topology of different types of wireless networks. Thus, by utilizing the tool of stochastic geometry, we aim to more accurately evaluate the potential performance gains brought about by tackling the aforementioned challenges. 

\subsection{Related Work}\label{S1.1}

LEO satellite communication has gained enormous research interests in recent years. In \cite{add1}, the authors creatively integrated refracting reconfigurable intelligent surface and relay techniques into satellite IoT systems, opening up a brand new path for solving the satellite direct-to-indoor communication issues. The authors of \cite{add2} broke through traditional research perspectives of multiple access in satellite networks, and studied massive access and interference management issues by exploiting rate splitting multiple access, which injects new vitality and directions into the satellite-IoT field. In the context of multi-altitude LEO satellite networks, the authors of \cite{add3} proposed a hybrid beamforming design for terrestrial users that are assumed to employ holographic metasurfaces. It involves optimizing both holographic beamformer and digital beamformer, and for the latter a low-complexity minimum mean square error (MMSE) beamforming algorithm is designed by exploiting the stochastic geometry based distribution of the LEO satellite constellation. 

However, research on stochastic geometry based performance analysis of IoT-over-LEO satellite systems, particularly on the uplink, remains limited, with only a few studies such as \cite{yastrebova2020theoretical, al2021modeling, 9838778, 9676997, 10463093, 10679173}. Specifically, Yastrebova \emph{et al.} \cite{yastrebova2020theoretical} investigated the impact of terrestrial interference on LEO satellite uplinks in high international mobile telecommunications (IMT) frequency bands. The authors of \cite{9676997} put forth an analytical framework that captures the frame repetition behavior based on the LOS probability and analyzed the coverage performance as a function of the frame success rate. The work \cite{al2021modeling} analyzed uplink coverage performance of the cooperative satellite-terrestrial network (CSTN) in which the connection between users and satellites is realized through base station (BS) relaying. Chan \emph{et al.} \cite{9838778} analyzed the uplink performance of large-scale IoT-over-LEO satellite systems and, through simulations, compared the performance of random constellations with that of Walker constellations. The study \cite{10463093} compared direct and gateway-assisted uplinks, and optimized coverage and battery efficiency, while the authors of \cite{10679173} analyzed the uplink performance for nomadic communication with weak satellite coverage.

However, the studies \cite{yastrebova2020theoretical, al2021modeling, 9838778, 9676997, 10463093} modeled terrestrial terminals using a PPP or a Poisson cluster process (PCP), which assumes an infinite distribution region. Although the study \cite{10679173} considered a finite-size region for terrestrial devices and employed binomial point process (BPP) for modeling, it simplified the analysis by treating the terrestrial region as a large disk rather than a spherical cap, thereby neglecting the effect of Earth's curvature. This leads to significant differences between the work in \cite{10679173} and that of this paper, not only in the solution of distance distribution but also in the derivation process of  coverage probability. Moreover, the work in \cite{10679173} fails to take into account large-scale LEO constellations and the limited satellite coverage, both of which are the problems this paper aims to address. Thus, a significant gap exists in the current research landscape regarding the uplink performance analysis of realistic satellite-terrestrial networks with terrestrial terminals located in limited-size regions, particularly when also accounting for the impact of Earth's curvature. Filling this gap is the first motivation for our research.

PPP and PCP can be used to model the distributions of IoT devices in infinite space, while BPP is suitable for finite space \cite{Fundamental}. More specifically, although PPP and PCP are widely adopted, they inherently assume a random number of points, which contradicts the requirement for a fixed number of points distributed within a finite region in many IoT settings. The same issue also arises with the Matérn cluster process (MCP), since it is a special type of PCP and the parent points in MCP are constructed based on the rules of PPP. Therefore, PPP and PCP (and also MCP) are not suitable for modeling remote areas where terrestrial communication infrastructure is lacked and a fixed number of IoT devices are required to be concentrated in certain limited-size geographical regions, such as mountains or deserts.
By contrast, BPP is particularly suited to model scenarios where the number of devices is fixed and the points are distributed within a limited-size region. Therefore, the use of BPP can provide a model that better meets the realistic operational constraints when conducting performance analysis, making the analytical results more closely aligned with actual conditions.

Besides, the work \cite{yastrebova2020theoretical} used Rician fading, the studies \cite{al2021modeling, 9838778, 9676997} adopted an empirical model from \cite{9257490} for satellite-to-ground links in urban areas, and the work \cite{10463093} used a shadowed-Rician (SR) fading model which was further approximated as a Gamma function. The authors of \cite{10679173} employed the Nakagami fading model to analyze ground-to-air links. However, its model assumes a flat disk, neglecting Earth's curvature, which limits its applicability to the analysis of ground IoT devices' links to large-scale satellite constellations.
Therefore, none of these works has provided an analytical framework for IoT-over-LEO satellite systems based on the Nakagami fading model, which offers greater flexibility by allowing parameter adjustments to represent Rayleigh, Rician, and empirical fading models. Addressing these overlooked aspects in existing studies serves as our second motivation for this research.

Compared with the limited research on uplink performance, there is a relatively larger body of works focusing on downlink analysis \cite{2021Stochastic, 2020An, okati2021modeling, 7869087, 8068989, 9678973}. Even so, among these works, only the authors of \cite{9678973} considered the limited satellite coverage, while the rest assumed the satellite antenna to be omnidirectional. However, the work \cite{9678973} overlooked the presence of interference, which, while simplifying the derivation, deviates from real-world scenarios. When interference pattern and limited satellite coverage are jointly considered, they introduce significant complexity in performance analysis. Hence, it is clear that there is a lack of thorough performance analysis of the downlink, accounting for both limited satellite coverage and interference pattern.  This stands as our third motivation.

More importantly, existing studies on satellite communication analysis tend to focus either on the uplink \cite{yastrebova2020theoretical, al2021modeling, 9838778, 9676997, 10463093, 10679173} or on the downlink \cite{2021Stochastic, 2020An, okati2021modeling, 7869087, 8068989, 9678973}, overlooking the importance of holistic end-to-end (E2E) analysis. Although the authors of  \cite{song2022cooperative} discussed E2E performance, they did not considered any of these constraints --- finite-size terrestrial device distribution area, limited satellite coverage, and multiple satellites. Additionally, it lacks an analysis of distance distribution under these realistic operational constraints. Such analysis is crucial for understanding the holistic performance of IoT-over-LEO satellite systems, as it accounts for the interaction between different components and accurately evaluates the overall efficiency and reliability of data transmission from IoT devices to satellites and further to Earth stations (ESs). This significant oversight and the pressing need for a more comprehensive and realistic E2E analysis, constitute our fourth motivation. 

\begin{table*}[!tp]
	\caption{Comparison of representative state-of-the-arts satellite related networks with our work.} 
	\label{Comparison} 
	\vspace*{-8mm}
	\begin{center}
		\resizebox{1\linewidth}{!}{
			\begin{threeparttable}
				\begin{tabular}{c|c|c|c|c|c|c|c|c} 
					\bottomrule
					\textbf{Reference} & \textbf{\makecell{Link \\types}} & \textbf{\makecell{Channel fading model}} & 	\textbf{Interference} &\textbf{\makecell{Finite-size ground\\ device distribution area}} & \textbf{\makecell{Limited satellite \\coverage}} & \textbf{Earth curvature}& \textbf{Composite E2E} & \textbf{Point process} \\ \bottomrule
					
					Yastrebova \emph{et al.} \cite{yastrebova2020theoretical} & T-S & Rician & Rician & - & - & \checkmark & -	&\makecell{S: PPP \\T: PPP} \\ \hline\rule{0pt}{8pt}
					
					Manzoor \emph{et al.}	\cite{9676997} & T-S & empirical model \cite{9257490}& empirical model \cite{9257490} &- & - & \checkmark& -	 & \makecell{S: One Satellite \\T: PPP} \\ \hline\rule{0pt}{8pt}
					
					Homssi \emph{et al.} \cite{ al2021modeling} & T-S & empirical model \cite{9257490}& empirical model \cite{9257490} & - & - & \checkmark& -	 & \makecell{S: PPP \\T: PPP} \\ \hline\rule{0pt}{8pt}
					
					Chan \emph{et al.} \cite{9838778} & T-S & empirical model \cite{9257490}& empirical model \cite{9257490} & - & - & \checkmark& -	 & \makecell{S: PPP \\T: PPP} \\ \hline\rule{0pt}{8pt}
					
					Talgat \emph{et al.} \cite{10463093} &  T-S/T-GW-S & \makecell{T-S: SR (Approximating \\ to a Gamma function) \\ T-GW-S: Rayleigh}& \makecell{T-S: SR (Approximating \\ to a Gamma function) \\ T-GW-S: Rayleigh} & - & \checkmark & \checkmark& -	 & \makecell{S: BPP \\T: PCP} \\ \hline\rule{0pt}{8pt}
					
					Dong \emph{et al.} \cite{10679173} & T-A-S & \makecell{T-A: Nakagami\\A-S: SR} & \makecell{T-A: Nakagami\\A-S: SR} & \checkmark & - & -& -	 & \makecell{T: BPP \\A: MHCPP} \\ \hline\rule{0pt}{8pt}	
					
					Talgat \emph{et al.} \cite{2021Stochastic} & S-GW-U & \makecell{S-GW: SR \\ GW-U: Rayleigh}& - & -& - & \checkmark& - & BPP/PPP\\ \hline\rule{0pt}{8pt}
					
					Al-Hourani \cite{2020An} & S-T & GMM  & GMM & - & - & \checkmark& -	 & PPP \\ \hline\rule{0pt}{8pt}
					
					Okati\emph{et al.} \cite{okati2021modeling} & S-T & Rayleigh  &- & - & - & \checkmark& - & NPPP \\ \hline\rule{0pt}{8pt}
					
					Sellathurai \emph{et al.} \cite{7869087} & S-T & \makecell{SR (Approximating \\ to a Gamma function)} & \makecell{SR (Approximating \\ to a Gamma function)} &- & - & \checkmark & - & PPP \\ \hline\rule{0pt}{8pt}
					
					Kolawole \emph{et al.} \cite{8068989} & S-U/GW-U & \makecell{S-U: SR (Approximating \\ to a Gamma function) \\ T-U: Nakagami} & Nakagami &- & - & \checkmark & - & PPP \\ \hline\rule{0pt}{8pt}
					
					Jung \emph{et al.} \cite{9678973} & S-T & SR  &- & - & \checkmark & \checkmark & - & BPP \\ \hline\rule{0pt}{8pt}
						Song \emph{et al.} \cite{song2022cooperative} & A-S-GW& \makecell{A-S: Nakagami \\ S-GW: SR} & Nakagami &- & - & - &\checkmark&	\makecell{S:One satellite \\ A:MHCPP} \\
							\hline\rule{0pt}{8pt}
					Our work  & T-S-ES & \makecell{T-S: Nakagami\\S-ES: SR} &  \makecell{T-S: Nakagami\\S-ES: SR} & \checkmark & \checkmark& \checkmark & \checkmark& \makecell{S: BPP\\T: BPP} \\

					 \bottomrule
					
				\end{tabular}
				\begin{tablenotes}
					\footnotesize
					\item S: satellite, A: aerial node, T: terrestrial node, ES: Earth station, GW: gateway. SR: shadowed-Rician fading, GMM: Gaussian mixture model. \checkmark: considered, - : not considered.
					\item PCP: Poisson cluster process (contains Mat{\'e}rn cluster process and Tomas cluster process),  PPP: Poisson point process, BPP: binomial point process, MHCPP: Mat{\'e}rn hard-core point process.
				\end{tablenotes}
			\end{threeparttable}
		}
	\end{center}
	\vspace*{-6mm}
\end{table*}

\subsection{Contributions of Our Paper}\label{S1.2}

Inspired by the insights gained from reviewing prior works, in this paper we focus on the performance analysis of IoT-over-LEO satellite systems under realistic operational constraints, such as finite-size ground area, limited satellite coverage, sophisticated Earth curvature, integral uplink and downlink, and link-specific interference. Specifically, we introduce an analytical framework for IoT-over-LEO satellite systems, where IoT devices are distributed within a limited area and satellite beams have restricted ranges. We derive an expression to characterize the communication coverage of IoT devices within this constrained space, where the devices establish connections with the ES through satellite relays.  In our exploration, we partition the system into two main segments: the IoT-device-to-satellite (T-S) link and the satellite-to-ES (S-ES) link. Based on the position distribution characteristics of the IoT devices and the ES, we analyze each segment using Nakagami and SR fading models, respectively. Our contributions can be summarized as follows.

\begin{itemize}[label=\textbullet,font=\Large]
	\item We introduce a novel stochastic geometry based model to describe the impact of IoT device distribution within a finite ground area on the performance of IoT-over-LEO satellite systems. Existing papers on the performance analysis of IoT-over-satellite systems typically assume a ground area of infinite size, which is unrealistic. In contrast, we model the terrestrial distributions of IoT devices using BPP, which allows for characterizing a fixed number of devices distributed within a given region. Furthermore, unlike prior works that assume an omnidirectional or Earth-tangent satellite coverage, we incorporate a practical beam-dependent limited satellite coverage, enabling a more accurate derivation of geometric relationships.
	
	\item Based on the above new modeling assumptions, we analyze the distance distributions of both the T-S link and the S-ES link under multiple real-world constraints, including the finite size of the IoT device distribution area, the curvature of the Earth, and the beam-dependent limited satellite coverage. Specifically, in existing studies based on PPP or PCP, the distance distribution function can be directly obtained through the probability of the empty event. However, when simultaneously considering the limited ground area, the Earth's curvature, and the limited satellite coverage, it is necessary to take into account the spatial geometric relationships. This significantly increases the difficulty of the derivation.

	\item We take into account different propagation environments of the uplink and downlink channels, and adopt distinct channel models that are appropriate for them. Additionally, the relevant interference patterns are considered in our theoretical analysis. For the T-S link, we adopt the Nakagami fading model, which serves as a more universal and adaptable framework for characterizing small-scale fading. In other words, it can offer greater flexibility by allowing parameter adjustments to represent Rayleigh, Rician, and empirical fading models. For the S-ES link in LEO satellite systems, where Earth stations may be located in urban areas with building obstructions, the SR model is well-suited as it captures the combined effects of shadowing and dominant LOS paths.
	
	\item We carry out the E2E performance analysis for an IoT-over-LEO satellite system in terms of the CP of T-S-ES links, and numerical simulations have validated our theoretical analysis. Specifically, we perform  an in-depth study on the impact of key system parameters, including the size of IoT device distribution area, the width of satellite beams, the numbers of IoT devices and satellites, and the antenna gains. Important insights are obtained from the extensive results. To the best of our knowledge, the E2E performance analysis of an IoT-over-LEO satellite system  where IoT devices communicate with the ES via satellites, has not been reported before.
\end{itemize}		

Table~\ref{Comparison} compares our work with some representative state-of-the-art contributions related to IoT-over-LEO satellite systems, in order to highlight the novelty of our contributions.

\begin{table}[!b]
	\vspace*{-6mm}
	\small
	\caption{Summary of Mathematical Notations adopted.} 
	\label{Table-Notation} 
	\vspace*{-7mm}
	\begin{center}
		\begin{tabular}{c|c}	
			\toprule
			\textbf{Notation} & \textbf{Description} \\ \midrule
			$\mathcal{A}_{\textrm{T}}$ & Finite-size IoT device distribution area \\
			$\mathcal{A}_{\textrm{C}}$ & Satellite coverage\\
			$\mathcal{A}_{\textrm{S}}$ &  Distribution region of visible satellites from ground\\
			$\theta_1$, $\theta_3$, $\theta_4$ & Earth-centered zenith angles of $\mathcal{A}_{\textrm{T}}$, $\mathcal{A}_{\textrm{C}}$, and $\mathcal{A}_{\textrm{S}}$ \\
			$\theta_2$ & Angle between the satellite and the target IoT device \\
			$R_{\textrm{c}}$ & Radius of $\mathcal{A}_{\textrm{T}}$\\
			$H$  & Height of satellites\\
			$r_{\textrm{e}}$ &  Radius of the Earth \\	
			$N_{T\textrm{-}S}$ &  Nakagami fading parameter \\
			$\varphi_{\textrm{s}}$ & Antenna beamwidth of satellites \\
			$\varphi$  &  \makecell{ Satellite-centered angle between the center\\ of the Earth and the ground transmitter} \\
			$N_{\textrm{T}}$ & Total number of IoT devices\\
			$N_{\textrm{S}}$ & Total number of satellites \\
			$\overline{D}$ & Duty cycle of IoT devices \\	
			$p_{\textrm{t}}$ & Power of IoT devices\\
			$p_{\textrm{s}}$ & Power of target satellite\\
			$p_n$ &  Power of interfering satellites \\ \bottomrule
		\end{tabular}
	\end{center}
	\vspace*{-2mm}
\end{table}

\subsection{Organization of the Paper and Notations}\label{S3.1}

The remainder of this paper is organized as follows. Section~\ref{S2} presents the topology model, the channel model, and the signal-to-interference-plus-noise ratio (SINR) model of the IoT-over-LEO satellite system considered. The distance distributions for both the target and interfering links of the T-S and S-ES links are derived in Section~\ref{S3}. Section~\ref{S4} provides our primary performance analysis results, which involve the derivations of the analytical CPs for the T-S link (i.e., uplink) and the S-ES link (i.e., downlink), respectively. Section~\ref{S5} presents the derivations of the analytical AERs for the T-S link and the S-ES link, respectively. In Section~\ref{S6}, we provide numerical results to verify our theoretical derivations and to quantify the impact of key system parameters. Our conclusions are drawn in Section~\ref{S7}.

Throughout this paper, $\mathbb{P}(\cdot)$ and $\mathbb{E}[\cdot]$ stand for the probability measure and the expectation operator, respectively. The Laplace transform of random variable $X$ is defined by $\mathcal{L}_X(s)=\mathbb{E}\,\left[\,\exp(-s X)\,\right]$. The cumulative distribution function (CDF) and probability density function (PDF) of random variable $X$ are denoted by $F_X(x)$ and $f_x(x)$, respectively. $\Gamma(\cdot)$ is the Gamma function, and the Pochhammer symbol is defined as ${\rm Ps}(x)_{n}=\Gamma(x+n)/\Gamma(x)$. The lower incomplete Gamma function is defined as $\gamma(a,x)=\int_0^xt^{a-1}\exp(-t)dt.$ $\binom nk$ denotes the binomial coefficient. ${_1F_1}\left(\cdot;\cdot;\cdot\right)$ is the confluent hypergeometric function of the first kind. For easy reference, Table~\ref{Table-Notation} lists the key symbols used in this paper.

\begin{figure}[t]\setcounter{figure}{0}
	\vspace*{-1mm}
	\begin{center}
		\includegraphics[width=1\columnwidth]{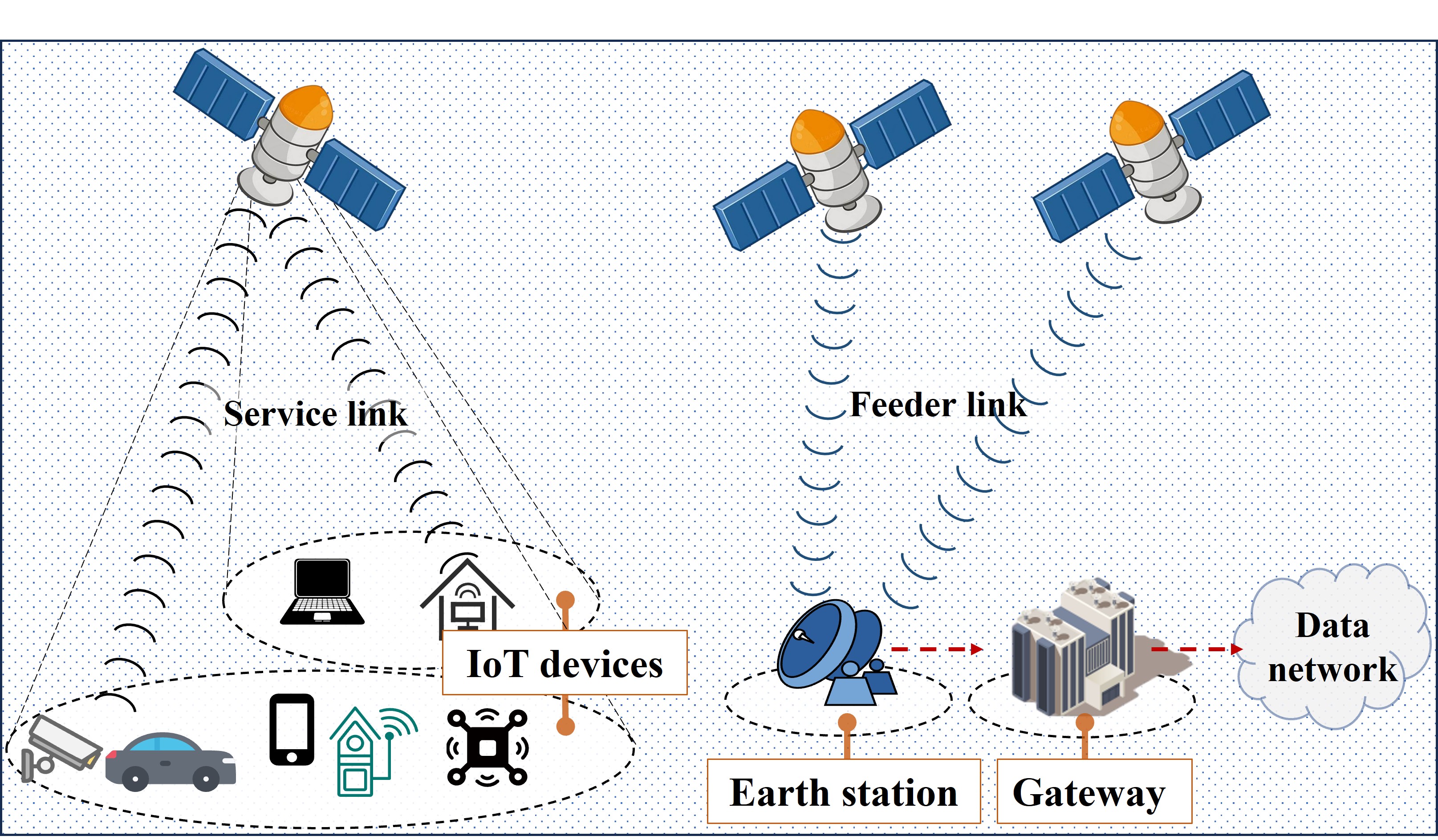}
	\end{center}
	\vspace*{-4mm}
	\caption{Illustration of the IoT-over-LEO satellite system.}
	\label{fig:1}
	\vspace*{-4mm}
\end{figure}

\section{System Model}\label{S2}

In this section, we present a holistic E2E system model for the IoT-over-LEO satellite system, which takes into account both the uplink and downlink, the finite-size ground area, the limited satellite coverage, the Earth curvature, and the link-dependent interference. Based on this model, we carry out the system performance analysis. 

\begin{figure*}[!b]\setcounter{figure}{1}
	\vspace*{-4mm}
	\begin{minipage}[t]{0.49\linewidth} %
		\centering
		\subfigure[service link] {	
			\centering
			\includegraphics[width=0.85\linewidth]{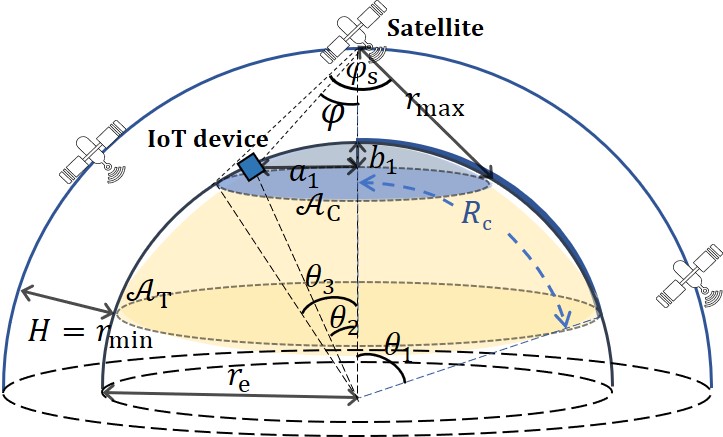}
			\label{fr1}
		}
	\end{minipage}  
	\begin{minipage}[t]{0.49\linewidth} %
		\centering
		\subfigure[feeder link] {	
			\centering
			\includegraphics[width=0.85\linewidth]{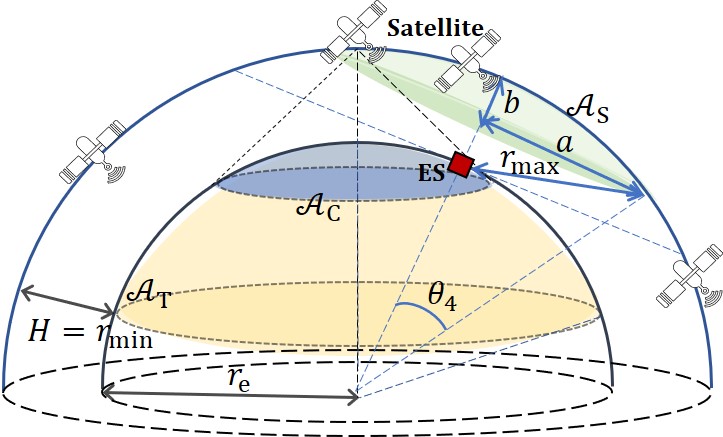}
			\label{fr2}
		}
	\end{minipage}
	\vspace*{-2mm}
	\caption{Illustration of the system's geometric relationships, where the yellow, blue and green spherical caps denote the finite-size ground area $\mathcal{A}_{\textrm{T}}$, satellite coverage  $\mathcal{A}_{\textrm{C}}$ and distribution region of visible satellites $\mathcal{A}_{\textrm{S}}$, respectively.}
	\label{fr} 
	\vspace*{-1mm}
\end{figure*}

\subsection{Geometric Model}\label{S2.1}  

As illustrated in Fig.~\ref{fig:1}, the system consists of multiple IoT devices ($T_n$, $n\! \in\! [1,\cdots,N_{\textrm{T}}]$, $N_{\textrm{T}}\! \geq\! 1$), a group of satellites ($S_l$, $l\! \in\! [1,\cdots,N_{\textrm{S}}]$, $N_{\textrm{S}}\! \geq\! 1$), and an ES. IoT devices are situated on the surface of the Earth, which is modeled as a perfect sphere with radius $r_{\textrm{e}}\! \approx\! 6371$\,km, while the satellite constellation is  placed on low circular orbits with the same altitude $H$. IoT devices connect to the core network using a transparent transmission architecture via relay satellites, the ES and the gateway, with the T-S uplink referred to as the service link and the S-ES downlink as the feeder link. We make the assumption that the devices maintain a consistently averaged duty cycle denoted as $\overline{D}$ during transmission. 

The system's geometric relationships are illustrated in Fig.~\ref{fr}. Since BPP is accurate for characterizing the nodal distributions in finite-area networks, $N_{\textrm{T}}$ IoT devices are distributed within a circular finite-size area, following a BPP distribution, denoted as $\Phi_{\textrm{T}}$\footnote{Since we use low Earth orbit satellites with a limited beam angle, satellite coverage  is relatively small. For instance, under our simulation parameters, the ground radius of this area is less than 90\,km, which means that IoT devices can be counted and modeled using BPP.}. This finite-size area takes the shape of a spherical cap, denoted as $\mathcal{A}_{\textrm{T}}$, with the arc $R_{\textrm{c}}$ of the largest circle on Earth’s surface and an Earth-centered zenith angle of $\theta_1=\frac{R_{\textrm{c}}}{r_{\textrm{e}}}$.  Our modeling and derivation are equally applicable to other propagation scenarios that can be properly characterized by Nakagami fading. Due to the limited number of satellites covering a given finite region, $N_{\textrm{S}}$ satellites are modeled by a BPP\,\footnote{LEO constellations currently being researched and developed are large, consisting of many satellites.  From the user’s perspective, the positions of the satellites can be approximated as independent. Furthermore, when satellites move along random circular orbits, this independence in their positions remains constant over time. Besides, the authors of [12] have further examined the accuracy of the BPP in modeling LEO constellations, focusing on its alignment with other commonly used models, such as Fibonacci lattice-based point set and orbit model-based point process. Therefore, using the BPP to model large-scale constellations is a reasonable assumption.}, denoted as $\Phi_{\textrm{S}}$, on a sphere with an altitude of $r_{\textrm{e}}+H$. 

We assume that each satellite is equipped with multiple directional beam antennas, which cooperate to form a coverage area directed towards the  Earth's center\,\footnote{Note that this paper does not model a single-beam system. In practice, LEO constellations typically employ multi-beam cooperation, and virtual beam technology is used to manage multiple beams, enabling coverage of the Earth's surface via a spherical cap region.}. To facilitate more flexible communication with satellites in various directions,  IoT devices employ directional antennas with steerable beams. For mathematical simplicity, we assume that both satellites and IoT devices display sectorized beam patterns and disregard the influence of the satellite's side lobes, given that the main lobe's directionality is of greater significance\,\footnote{These simplified beam patterns are commonly used in theoretical analyses that employ stochastic geometry \cite{7880676, 8335329, 9200666}. Although this assumption is not theoretically exact, it provides a sufficiently accurate approximation for most practical applications, as evidenced by the technical specifications of commercial antenna products used in cellular base stations.}.  We represent the main lobe antenna gain of the satellite $S$ as $G_{S}(\varphi_{\textrm{s}})$, which is determined based on the satellite’s beamforming configuration. Specifically, the coverage region is approximated as a spherical cap, and the gain is expressed as a function of the satellite coverage angle $\varphi_{\textrm{s}}$, given by the following equation \cite{9838778}:
\begin{align}\label{eqSmb} 
	G_{S}(\varphi_{\textrm{s}}) &= \frac{2}{1-\cos\big(\frac{\varphi_{\textrm{s}}}{2}\big)}.
\end{align} 
The main lobe and side lobe gains of IoT devices are denoted as $G_{\rm t}$ and $g_{\rm t}$, respectively, with the threshold angle between the main and side lobes in the beam pattern represented by $\varphi_{\textrm{t}}$. Furthermore, to ensure the timely reception of signals from the relay satellite, we consider the use of an omnidirectional antenna with a gain of $G_{ES}$ at the ES.

As depicted in Fig.~\ref{fr1}, the ground area covered by the satellite signal is represented by a blue spherical cap, denoted as $\mathcal{A}_{\textrm{C}}$, with an Earth-centered zenith angle $\theta_3$, while the finite-size area is indicated by a yellow spherical cap $\mathcal{A}_{\textrm{T}}$. Only IoT devices within the satellite coverage  can communicate with the satellite, and these devices are denoted as the set $\Phi_{\textrm{C}}$. Every IoT device is associated with a satellite referred to as the serving satellite in the subsequent context, which is positioned closest to its local zenith. Due to the limited number and signal coverage of satellites, there is a possibility that no satellite exists within the communication range of a target IoT device. We denote this probability as $P_0$ and provide a detailed discussion in Section \ref{S4.1}. To streamline the mathematical calculations, the device’s local zenith is represented as a point directly above it, aligned with the Earth-centered vertical line and lying on the satellite's orbital sphere. We further assume that these IoT transmitters employ code-division multiple access (CDMA) and share a common frequency $f_1$. The same access technique is utilized in the feeder link, with satellites operating on a different frequency $f_2$. Consequently, the signals emitted by IoT devices do not interfere with the reception at the ES. But the received signal at the serving satellite suffers the multi-user interference (MUI) arising from other IoT devices within the coverage area of the satellite beam on the service link.

Since the satellite antenna beam is a fixed directional beam with an angle $\varphi_{\textrm{s}}$, wireless transmissions propagate to the ES exclusively from the satellites positioned within the area composed of satellites that can cover the ES in their coverage areas. This communicable area is called the distribution region of visible satellites, which is represented by the green spherical cap $\mathcal{A}_{\textrm{S}}$ in Fig.~\ref{fr2}. The set of satellites within the distribution region of visible satellites is denoted as  $\Phi_{\textrm{L}}$. It can be seen that $\mathcal{A}_{\textrm{S}}$'s Earth-centered zenith angle, denoted as $\theta_4$, is equal to $\theta_3$, i.e., $\theta_4\! =\! \theta_3$. The satellite's height $H$ is equivalent to the minimum distance $r_{\textrm{min}}$ to the ES, achieved when it is directly overhead, i.e., $H\! =\! r_{\textrm{min}}$. The maximum possible distance that a satellite can communicate with the ES  is denoted by  $r_{\textrm{max}}$, which will be discussed in Section~\ref{S3}. This maximum distance is realized when the satellite-centered angle between the center of the Earth and the ground transmitter, $\varphi$, is equal to the half of $\varphi_{\textrm{s}}$, i.e., $\varphi\! =\! \frac{1}{2}\varphi_{\textrm{s}}$. 

\subsection{Service Link Channel Model}\label{S2.2}

Let $m$ be the index of the given target node, $T_m$ be the target IoT device, $T_n$ be the $n$th interfering IoT device which is within the coverage of the serving satellite and imposes interference on the reception of $T_m$'s signal at the serving satellite, and $D_{T\textrm{-}S}$ be the overall directional gain from IoT device $T$ to satellite $S$. The value of $D_{T_m\textrm{-}S}$ for the T-S link from $T_m$ to satellite $S$ is given by $D_{T_m\textrm{-}S}\! =\! G_{\rm t} G_{S}(\varphi)$. For any interfering link, we assume that the angle of arrival and the angle of departure of the signal are independently and uniformly distributed in the range $(0, 2\pi]$, which results in a random directivity gain denoted as $D_{T_n\textrm{-}S}$. The probability distribution for $D_{T_n\textrm{-}S}$ can be expressed as:
\begin{align}\label{eqDGi} 
	D_{T_n\textrm{-}S} =& \left\{ \begin{array}{cl}
		G_{\rm t}G_{S}(\varphi_{\textrm{s}}), & {P}_{{\rm Mb},{\rm Mb}}=\frac{\varphi_{\textrm{t}}}{2\pi} , \\ 
		\,g_{\rm t}\, G_{S}(\varphi_{\textrm{s}}), & {P}_{{\rm Sb},{\rm Mb}}=1-\frac{\varphi_{\textrm{t}}}{2\pi} ,
	\end{array} \right.
\end{align} 
where ${P}_{tx,{\rm Mb}}$ denotes the probability of T-S link in state $(tx,{\rm Mb})$. Here the first subscript $tx \in\{{\rm Sb},{\rm Mb}\}$ represents the beam state of the IoT device and the second subscript $\rm Mb$ represents the beam state of the satellite, while $\rm Sb$ and $\rm Mb$ denote the side lobe and main lobe, respectively.

Let $R_{\textrm{T}}$ be the distance between the IoT device and the serving satellite, $\alpha_1$ be the path-loss exponent of the service link, and $f_1$ be the carrier frequency. Denote the speed of light as $c$. Then the path loss of the service link is given by \cite{9678973}:
\begin{align}\label{eqPLsl} 
	l_1(R_{\textrm{T}}) =& \Big( \frac{c}{4\pi f_1}\Big)^2 R_{\textrm{T}}^{-\alpha_1} .
\end{align}

The satellite-terrestrial connection commonly relies on LOS transmission, but it may encounter obstacles such as buildings or plants which impede its signal propagation \cite{9520123}. In practice, many IoT devices are situated in outdoor environments prone to rich scattering and/or LOS paths, hence the signals from IoT devices to satellites often propagate through multiple paths characterized by the Nakagami fading model, which offers greater flexibility by allowing parameter adjustments to model Rayleigh, Rician, and empirical fading\,\footnote{For instance, with the Nakagami fading parameter set to $N_{T\textrm{-}S}=1$ and $N_{T\textrm{-}S}=\frac{(K+1)^2}{2K+1}$, it resorts to the Rayleigh and Rician-K distributions, respectively. Furthermore, by tuning its parameter $N_{T\textrm{-}S}$, it is possible to model the signal fading conditions spanning from severe to moderate, while making the distribution fit to empirically measured fading data sets \cite{Park}.}. Since the channel model $h_{T\textrm{-}S}$ exhibits Nakagami fading, $|h_{T\textrm{-}S}|^2$ can be modeled as a random variable that follows a normalized Gamma distribution. For the sake of simplicity, it is assumed that the Nakagami fading parameter $N_{T\textrm{-}S}$ is a positive integer \cite{6932503}. 
We further posit that the Doppler shifts arising from the high-speed movement of LEO satellites can be successfully mitigated using accurate estimation methods. These methods leverage detailed ephemeris data of the satellites, such as orbit types, altitudes, positions, and velocities, which can be precisely predicted in advance \cite{7468478,9107497}.

Building upon the aforementioned modeling, the SINR at the receiving satellite for the signal originating from the target IoT device $T_m$ can be expressed as:
\begin{align} 
	\mathrm{SINR}_1 =& \frac{p_{\textrm{t}} D_{T_m\textrm{-}S}\left|h_{T_m\textrm{-}S}\right|^2 l_1\big(R_{T_m}\big)}{I_{\textrm{T}}+\sigma^2} \label{eqSINR0} \\
	\approx&	\frac{p_{\textrm{t}} D_{T_m\textrm{-}S}\left|h_{T_m\textrm{-}S}\right|^2 l_1(R_{T_m})} {I_{\textrm{T}}} \label{eqSINR},
\end{align}
where $p_{\textrm{t}}$ is the transmit power at IoT devices, $R_{T_m}$ is the distance from node $T_m$ to the satellite $S$, and $\sigma^2$ is the strength of additive white Gaussian noise (AWGN) of both the T-S link and S-ES link, while the interference power $I_{\textrm{T}}$ is given by  
\begin{align}\label{eqInterf} 
	I_{\textrm{T}} &= \sum_{T_n\in\Phi_{\textrm{C}} {\backslash} T_m} p_{\textrm{t}}D_{T_n\textrm{-}S} \overline{D}\left|h_{T_n\textrm{-} S }\right|^2 l_1({R_{T_n}}), 
\end{align}
in which $\overline{D}$ is the averaged duty cycle, $T_n$ are interfering IoT devices, and $R_{T_n}$ is the distance from node $T_n$ to $S$. The approximation in (\ref{eqSINR}) is due to the fact that the system is interference limited and $I_{\textrm{T}} \gg \sigma^2$.

\subsection{Feeder Link Channel Model}\label{S2.3}

Let $S_m$ be the last-hop satellite that transmits the signal to the ES, $S_n$ be the interfering satellite, and $G_{S\textrm{-}ES}$ be the overall beam gain from satellite $S$ to the ES. Since the satellite's directional antenna utilizes fixed beams with negligible side lobe gain and the ES employs an omnidirectional antenna, the value of the beam gain remains uniform for both the given target satellite $S_m$ and the interfering satellite $S_n$, i.e., $D_{S_m\textrm{-}ES} = D_{S_n\textrm{-}ES}= G_{S}(\varphi_{\textrm{s}}) G_{ES}$.

The path-loss of the feeder link between the satellite and the ES is given by
\begin{align}\label{eqPLfl} 
	l_2(R_{\textrm{S}}) =& \Big(\frac{c}{4\pi f_2}\Big)^2 R_{\textrm{S}}^{-\alpha_2},
\end{align}
where $f_2$ is the carrier frequency of the feeder link, $R_{\textrm{S}}$ is the distance between the ES and its serving satellite, and $\alpha_2$ is the path-loss exponent of the feeder link.

ESs are commonly deployed in outdoor environments with low population density and minimal industrial activity, resulting in reduced electromagnetic interference. Additionally, the open spaces in suburban areas reduce the presence of obstacles that might affect signal transmission.   However, for the LEO systems considered, ESs are also often deployed in urban areas, where the S-ES link may be affected by shadowing from buildings and other obstructions. Hence, the S-ES link is characterized by using the SR fading model \cite{2003A}, which encompasses both the LOS shadowing component and the scattering component, usually used to analyze satellite-terrestrial links in different fixed and mobile satellite services operating in various frequency bands \cite{jung2018outage}. Therefore, we use the SR fading model to characterize the small-scale fading of the feeder link, with $2\bar{c}$ defining the mean power of the multipath component excluding the LOS component and $\Omega$ denoting the mean power of the LOS component. The Nakagami  fading parameter is denoted by $q$, and the small-scale fading is denoted by $|h_{S\textrm{-}ES}|^{2}$ in this study, where $h_{S\textrm{-}ES}$ denotes the channel coefficient of the S-ES link. Consequently, the PDF of $|h_{S\textrm{-}ES}|^{2}$ can be represented as \cite{zhang2019performance}:	
\begin{align}\label{eqSSF} 
	f_{|h_{S{\textrm{-}}ES}|^2}(x) =& \kappa\, \mathrm{exp}(-\beta x)\, {_1F_1}\left(q\, ;1\, ;\delta x\right ),
\end{align}
where $\kappa =\frac{(2 \bar{c} q)^q}{2 \bar{c}(2 \bar{c} q+\Omega)^q}$, $\delta=\frac{\Omega}{2 \bar{c}(2 \bar{c} q+\Omega)}$, and $\beta=\frac{1}{2 \bar{c}}$.

Similar to the service link, the SINR of the  signal received at the ES from the satellite $S_m$ is given by
\begin{align}\label{eqESsinr} 
	\mathrm{SINR}_2 &\approx\frac{p_{\textrm{s}} D_{S_m\textrm{-}ES}\left|h_{S_m\textrm{-} ES}\right|^2 l_2(R_{S_m})} {I_{\textrm{S}}},
\end{align}
where the interference power $I_{\textrm{S}}$ is given by
\begin{align}\label{eqInterf1} 
	I_{\textrm{S}} &= \sum_{S_n\in\Phi_{\textrm{L}} {\backslash} \{S_m\}} p_n D_{S_n\textrm{-}ES}\left|h_{S_n\textrm{-} ES}\right|^2 l_2(R_{S_n}), 
\end{align}
$p_{\textrm{s}}$ and $p_n$ are the transmit power at target satellite $S_m$ and interfering satellite $S_n$, respectively, while $R_{S_m}$ and $R_{S_n}$ are the distances from the target node $S_m$ and the interfering node $S_n$ to the ES, respectively. Owing to the interference limited nature, $I_{\textrm{S}} \gg \sigma^2$, and we ignore the impact of noise.

\section{ Distance Distribution}\label{S3}

To derive the CP expression, it is essential to first characterize key distance distributions arising from the stochastic geometry inherent in the considered system. Additionally, we delve into the success probabilities of the binomial process associated with the counts of terrestrial interfering IoT devices and space interfering satellites, respectively. The success probability refers to the likelihood of a point being situated within the area of interest. In the case of homogeneous BPPs, the success probability is calculated as the ratio of the surface area within the region of interest to the total surface area where all points are distributed \cite{1995Stochastic}.

Let $r$ be the distance between IoT device $T$ and satellite $S$, and $\theta_2$ signify the central angle between the two. From basic geometry, we obtain 
\begin{align} 
	& r \, \cos(\varphi) + r_{\textrm{e}} \,\cos(\theta_2) = H + r_{\textrm{e}} ,	\label{L0-1} \\
	& r^2 =\left((H + r_{\textrm{e}}) \sin(\theta_2) \right)^2 + \left( (H + r_{\textrm{e}}) \cos(\theta_2) - r_{\textrm{e}}\right)^2.	\label{L0-2}
\end{align} 	
Substituting (\ref{L0-1})	into (\ref{L0-2}) leads to the relationship between $r$ and $\varphi_{\textrm{s}}$ given by 
\begin{align}\label{L0-3} 
	r = (H + r_{\textrm{e}}) \cos(\varphi)-\sqrt{r_{\textrm{e}}^2\- (H+r_{\textrm{e}})^2 \sin^2(\varphi)}.	
\end{align} 
Then, we can obtain the maximum communication distance $r_{\textrm{max}}$ from satellite to IoT devices or the ES by substituting $\varphi=\frac{\varphi_{\textrm{s}}}{2}$ into (\ref{L0-3}):
\begin{align}\label{L0-4} 
	\!\!\!\!\!	r_{\textrm{max}}\! = \! (H + r_{\textrm{e}}) \cos\Big(\frac{\varphi_{\textrm{s}}}{2}\Big)\! -\! \sqrt{r_{\textrm{e}}^2\! -\! (H + r_{\textrm{e}})^2 \sin^2\Big(\frac{\varphi_{\textrm{s}}}{2}\Big)}.	
\end{align} 

\begin{Lemma}\label{L1}
	The CDF of the distance $R_{\textrm{S}}$ from any specific satellite $S$ to an IoT device  $T$ or to the ES is given by
	\begin{align}\label{L1-1} 
		F_{R_{\rm{S}}}(r_{\rm{s}}) \triangleq & \mathbb{P}\left( R_{\rm{S}} \leq r_{\rm{s}} \right) \nonumber \\
		=& \left\{ \begin{array}{cl}
			0	, & r_{\rm{s}} < H , \\ 
			\frac{ r_{\rm s}^2-H^2}{4(r_{\rm{e}}+H)r_{\rm{e}}}	, & H \leq r_{\rm{s}} \leq H+2r_{\rm{e}} , \\ 
			1	, & r_{\rm{s}} > H+2r_{\rm{e}} ,
		\end{array} \right.
	\end{align} 
	and the corresponding PDF is given by
	\begin{align}\label{L1-2} 
		f_{R_{\rm{S}}}(r_{\rm{s}}) = \left\{ \begin{array}{cl}
			\frac{ r_{\rm{s}} }{2(r_{\rm{e}}+H)r_{\rm{e}}}, & H \leq r_{\rm{s}} \leq H+2r_{\rm{e}} , \\
			0 , & \textrm{otherwise} .
		\end{array} \right.
	\end{align}
\end{Lemma} 

\begin{proof}
	See Appendix~\ref{ApA}.
\end{proof}

\begin{remark}\label{remark0}
To derive the distance distribution of the service link, we first need to determine the distance distribution from the target device to its nearest service satellite, followed by the distance distribution from interfering devices to this particular service satellite. The former requires obtaining the distance distribution from an arbitrary point on the satellite orbital sphere to a fixed point on the ground, which is given in Lemma \ref{L1}. The latter involves the distance distribution from an arbitrary point on the ground to a fixed point on the satellite orbital sphere, as presented in Lemma \ref{L2}. Additionally, the distance distribution of the feeder link is primarily based on Lemma \ref{L1}.
\end{remark}
Thus, the CDF of the distance $R_{T_m}$ between the target IoT device $T_m$ and the nearest satellite can be obtained as: 
\begin{align}\label{L1-3} 
	& F_{R_{T_m}}(r_m) = 1 - \mathbb{P}(R_{T_m} > r_m) = 1 - \prod_{S \in \Phi_{\textrm{S}}} \mathbb{P}(R_{\textrm{S}} >r_m) \nonumber \\
	&=\! \left\{\!\!\!\! \begin{array}{cl}
		0 \!\!\! &\!\!\!  , r_m < H, \\
		1-\left(1 - \frac{ r_m^2 - H^2}{4(r_{\textrm{e}} + H)r_{\textrm{e}}}\right)^{N_{\textrm{S}}} \!\!\! &\!\!\!  , H \leq r_m \leq H+2r_{\rm{e}} , \\ 
		1 \!\!\! & \!\!\!, r_m > H + 2r_{\rm{e}} .
	\end{array}\!\! \right.\!\!
\end{align} 
By differentiating (\ref{L1-3}), the corresponding PDF is given by
\begin{align}\label{L1-4} 
	& f_{R_{T_m}}(r_m) = \nonumber \\
	& \hspace*{2mm}\left\{\!\!\! \begin{array}{cl}
		N_{\textrm{S}}\! \left(\! 1\! -\! \frac{r_m^2-H^2}{4(r_{\textrm{e}}\! +\! H)r_{\textrm{e}}}\! \right)\! ^{N_{\textrm{S}}\! -\! 1}\! \frac{r_m}{2(r_{\textrm{e}}\! +\! H)r_{\textrm{e}}},  \!\! & \!\! H\! \leq \!r_m\! \leq\! H\! +\! 2r_{\rm{e}}, \\
		0 \!\! &\!\! \textrm{otherwise} . 
	\end{array}\!\! \right.\!\!
\end{align} 

\begin{Lemma}\label{L2}
	The CDF of the distance $R_{\rm{T}}$ from any specific IoT device  $T$ to a satellite $S$ is given as 
	\begin{align}\label{L2-1} 
		F_{R_{\rm{T}}}(r_{\rm{t}}) \triangleq & \mathbb{P}\left( R_{\rm{T}} \leq r_{\rm{t}} \right) \nonumber \\
		=& \left\{ \begin{array}{cl}
			0	, & r_{\rm{t}} < H , \\ 
			\frac{r_{\rm{t}}^2 - H^2}{4(r_{\rm{e}}+H)r_{\rm{e}} }	, & H \leq r_{\rm{t}} \leq H+2r_{\rm{e}} , \\ 
			1	, & r_{\rm{t}} > H+2r_{\rm{e}} .
		\end{array} \right.
	\end{align} 
\end{Lemma}

\begin{proof}
	See Appendix~\ref{ApB}.
\end{proof}

The PDF $f_{R_{\textrm{T}}}(r_\textrm{t})$ of the distance $R_{\rm{T}}$ from IoT device $T$ to satellite $S$ can be obtained by differentiating (\ref{L2-1}) as 	 
\begin{align}\label{L2-2} 
	f_{R_{\textrm{T}}}(r_\textrm{t}) =& \left\{  \begin{array}{cl}
		\frac{r_\textrm{t} }{2(r_{\textrm{e}}+H)r_{\textrm{e}} } , & H \leq r_\textrm{t} \leq H+2r_{\textrm{e}} , \\
		0 , & \textrm{otherwise} .
	\end{array} \right.
\end{align} 

Based on (\ref{L1-3}), the serving satellite can be determined, and naturally the interfering IoT devices are located within the coverage $\mathcal{A}_{\textrm{C}}$ of the serving satellite, i.e.,  $\Phi_{\textrm{C}}\! \neq\! \emptyset$. Hence, the interfering IoT devices must satisfy the constraint that the distance from these ground devices to the service satellite, i.e., $R_{\textrm{T}}$, is less than or equal to the maximum communication distance $r_{\textrm{max}}$, i.e., $R_{\textrm{T}} \leq r_{\textrm{max}}$. Thus, the CDF of the distance $R_{T_n}$ between interfering IoT devices and the serving satellite is defined as
\begin{align}\label{L2-3} 
	&F_{R_{T_n}}(r_{n}) = \mathbb{P}\left( R_{T_n}\leq r_{n} \right)= \mathbb{P}\left(R_{\textrm{T}} \leq r_{n}|\Phi_{\textrm{C}}\neq 0 \right) \nonumber\\
	&=\! \mathbb{P}\left( R_{\textrm{T}} \leq r_{n}| R_{\textrm{T}} \leq r_{\textrm{max}}  \right)\! =\! \frac{\mathbb{P}( R_{\textrm{T}} \leq r_{n}, R_{\textrm{T}} \leq r_{\textrm{max}})}{ \mathbb{P}(R_{\textrm{T}} \leq r_{\textrm{max}})} \nonumber\\
	&=\! \frac{\mathbb{P}( R_{\textrm{T}} \leq \textrm{min}(r_{n},r_{\textrm{max}}) )}{ \mathbb{P}(R_{\textrm{T}} \leq r_{\textrm{max}})},\!
\end{align} 
where $r_{\textrm{max}}$ is given by (\ref{L0-4}). Based on the comparison relationship between $r_{n}$ and $r_{\textrm{max}}$, and by applying Eq. (\ref{L2-1}), the CDF of the distance $R_{T_n}$ between interfering IoT devices and the serving satellite can be obtained as
\begin{align}\label{L2-4} 
	F_{R_{T_n}}(r_{n})=& \left\{ \begin{array}{cl}
		0	, & r_n < H , \\ 
		\frac{ F_{R_{\textrm{T}}}(r_n)}{F_{R_{\textrm{T}}}(r_{\textrm{max}}) }	, & H \leq r_n \leq r_{\textrm{max}} , \\ 
		1	, & r_n > r_{\textrm{max}} ,
	\end{array} \right.
\end{align} 
and the corresponding PDF is given by
\begin{align}\label{L2-5} 
	f_{R_{T_n}}(r_{n}) =& \left\{ \begin{array}{cl}
		\frac{ 2r_n	}{r_{\textrm{max}}^2-H^2} ,  & H \leq r_n \leq r_{\textrm{max}} , \\
		0	, &  {\textrm{otherwise}} .
	\end{array} \right.
\end{align} 

IoT devices connect to the ES via satellite relays. Since the distance covered by a single-hop satellite link often falls short for communication needs, signals are aggregated at one satellite and then sent through multi-hop inter-satellite links to reach the ES. It is assumed that there is at least one satellite within the   distribution region of visible satellite from the view of the ES. A satellite within this distribution area is randomly selected to link to the ES, and the remaining satellites within this distribution area are classified as interfering satellites. Thus, the distance $R_{S_m}$ between the target satellite $S_m$ and the ES is given as
\begin{align}\label{L2-6} 
	F_{R_{S_m}}(r_m') =& \mathbb{P}\left( R_{S_m}\leq r_{m}' \right) = \mathbb{P}\left( R_{\textrm{S}} \leq r_{m}'| R_{\textrm{S}} \leq r_{\textrm{max}} \right) \nonumber \\
	&=  \left\{ \begin{array}{cl}
		0	, & r_m' < H , \\ 
		\frac{ F_{R_{\textrm{T}}}(r_{m}')}{F_{R_{\textrm{T}}}(r_{\textrm{max}}) }	, & H \leq r_m' \leq r_{\textrm{max}} , \\ 
		1	, & r_m' > r_{\textrm{max}} .
	\end{array} \right.
\end{align} 
The corresponding PDF is given by
\begin{align}\label{L2-7} 
	f_{R_{S_m}}(r_{m}') =& \left\{ \begin{array}{cl}
		\frac{2r_m'}{r_{\textrm{max}}^2-H^2} ,  & H \leq r_m' \leq r_{\textrm{max}} , \\
		0 , & {\textrm{otherwise}} .
	\end{array} \right.
\end{align} 

Due to the fact that the distance variables $r_m'$ and $r_n'$ for the target satellite and interfering satellites are independently and identically distributed, $f_{R_{S_m}}(r_{m}')\! =\!	f_{R_{S_n}}(r_{n}')$.

\begin{remark}\label{Rm1}
	The derived distance distribution provides valuable insights into the spatial characteristics of IoT-satellite communication, including the impact of interfering IoT devices on the serving satellite as well as the distance distribution from the satellite to the ES. By considering finite device regions and satellite coverage limits, it helps to understand the communication links both between IoT devices and their serving satellites and between the satellite and the ES. This framework also accounts for interference from nearby devices, enabling more informed decisions in satellite constellation placement and coverage  design, particularly in scenarios with spatial constraints. Additionally, the framework serves as a flexible tool for analyzing other similar communication systems with  finite-size ground area, limited satellite coverage and interference factors.
\end{remark} 

\section{Coverage Probability}\label{S4}

We conduct CP analysis for both the T-S and S-ES links. The CP refers to the likelihood that the SINR at the receiver surpasses the minimum SINR threshold necessary for successful data transmission. In other words, if the SINR of the received signal at the receiver surpasses the threshold $T$, the transmitter is inside the coverage area of the receiver.

\subsection{CP over the T-S Link} \label{S4.1}

When the serving satellite is at distance $R_{T_m}$ which is greater than $r_{\textrm{max}}$ from the target IoT device, the closest satellite to the target IoT device is outside the communication range, which means that there is no service satellite, i.e.,
\begin{align}\label{eqS4-1-1} 
	& \mathbb{P}( \mathrm{SINR}_1 \geq T_1|R_{T_m} > r_{\textrm{max}}) = 0.
\end{align}
The probability of $R_{T_m}\! >\! r_{\textrm{max}}$, denoted as $P_0$, is given by
\begin{align}\label{P0} 
	P_0 =& \mathbb{P}({R_{T_m}>r_{\textrm{max}}}) = 1 -\mathbb{P}({R_{T_m}\leq r_{\textrm{max}}}) \nonumber \\
	=& 1 - F_{R_{T_m}}(r_{\textrm{max}}) = \left(1 - \frac{ r_{\textrm{max}}^2 - H^2}{4(r_{\textrm{e}} + H) r_{\textrm{e}}}\right)^{N_{\textrm{S}}}.
\end{align}

\begin{figure*}[!bp]\setcounter{equation}{33}
	\vspace*{-4mm}
	\hrulefill 
	\begin{align}\label{cov1} 
		P_{\mathrm{cov}}^{\textrm{T-S}}\! =& \left(1-\left(1-\frac{ r_{\textrm{max}}^2-H^2}{4(r_{\textrm{e}}+H)r_{\textrm{e}}}\right)^{N_{\textrm{S}}}\right) \sum_{n=1}^{N_{T\textrm{-}S}}(-1)^{ n+1} \binom{N_{T\textrm{-}S}}{n}
		\sum_{n_I=1}^{N_{\textrm{T}}-1}\! \binom{N_{\textrm{T}}-1}{n_I} \left(\frac{r_{\textrm{max}} \cos\big(\frac{\varphi_{\textrm{s}}}{2}\big) - H}{r_{\textrm{e}} - r_{\textrm{e}} \cos\big(\frac{R_{\textrm{c}}}{r_{\textrm{e}}}\big)}\right)^{n_I}\nonumber \\
		&\times \int_H^{r_{\textrm{max}}}\!\!\! \Bigg( \int_{H}^{r_{\textrm{max}}}\!\! \Bigg(\!\!\! \left(\!1\! + \!\frac{2n T_1r_m^{\alpha_1}\overline{D}}{(N_{T\textrm{-}S}!)^{\frac{1}{N_{T\textrm{-}S}}}r_n^{\alpha_1}\big(1-\cos\big(\frac{\varphi_{\textrm{s}}}{2}\big)\big)}\! \right)^{\!\! -N_{T\textrm{-}S}}\!\!\!\!\!\!\!\! \frac{\varphi_{\textrm{t}}}{2\pi} + \! \left(\! 1\! + \!\frac{2n T_1r_m^{\alpha_1}\overline{D} g_{\textrm t}}{(N_{T\textrm{-}S}!)^{\frac{1}{N_{T\textrm{-}S}}}r_n^{\alpha_1}G_{\textrm t}\big(1 - \cos\big(\frac{\varphi_{\textrm{s}}}{2}\big)\big)}\!\right)^{\!\! -N_{T\textrm{-}S}}\!\!\!\!\!\! \left(1-\frac{\varphi_{\textrm{t}}}{2\pi}\right)\! \Bigg) \nonumber \\ 
		& \times \frac{2r_n}{r_{\textrm{max}}^2\! -\! H^2} \mathrm{d} r_n\! \Bigg)^{n_I}\!\! \left(\! \frac{H\! +\! r_{\textrm{e}}\! -\! r_{\textrm{max}} \cos\big(\frac{\varphi_{\textrm{s}}}{2}\big) - r_{\textrm{e}} \cos\big(\frac{R_{\textrm{c}}}{r_{\textrm{e}}}\big)}{r_{\textrm{e}} - r_{\textrm{e}}\cos\big(\frac{R_{\textrm{c}}}{r_{\textrm{e}}}\big)}\right)^{N_{\textrm{T}}-1-n_I} \!\!\!\!\! N_{\textrm{S}} \left(1\! -\! \frac{r_m^2-H^2}{4(r_{\textrm{e}}\! +\! H)r_{\textrm{e}} }\right)^{\!\!N_{\textrm{S}}-1} \!\!\!\!\! \frac{r_m }{2(r_{\textrm{e}}\! +\!H)r_{\textrm{e}}} \mathrm{d}r_m. 	 
	\end{align} 
	\vspace*{-1mm}
\end{figure*}

We now formulate the CP under the Nakagami fading assumption for the service link.\setcounter{equation}{27}

\begin{Lemma}\label{L3}
	The Laplace transform of random variable $I_{\rm{T}}$ is 
	\begin{align}\label{eqLapTraI} 
		\mathcal{L}_{I_{\rm{T}}}(s)\! =& \mathbb{E}_{N_I,R_{T_n},D_{T_n\textrm{-}S}}\!\! \left[\! \prod_{T_n\in\Phi_{\rm{C}} {\backslash} T_m}\!\!\!\! \left(\!1\! +\! \frac{t_0 D_{T_n\textrm{-}S} }{N_{T\textrm{-}S}R_{T_n}^{\alpha_1}}\!\right)^{\!\!\!-N_{T\textrm{-}S}} \!\right]\!\! ,\!
	\end{align}
	with $t_0 = \frac{n\eta T_1r_m^{\alpha_1}\overline{D}}{D_{T_m\textrm{-}S}}$.
\end{Lemma}

\begin{proof}
	See Appendix~\ref{ApD}.
\end{proof}

The target IoT devices are located within the finite-size area $\mathcal{A}_{\textrm{T}}$. Once a target IoT device selects a service satellite, any other IoT devices within the satellite coverage $\mathcal{A}_{\textrm{C}}$ of the serving satellite are interfering devices. However, in reality, there exist IoT devices that are outside the finite-size area $\mathcal{A}_{\textrm{T}}$ but are also within the satellite coverage $\mathcal{A}_{\textrm{C}}$, hence becoming interfering nodes. Thus, we assume that when the satellite coverage  $\mathcal{A}_{\textrm{C}}$  partially intersects with $\mathcal{A}_{\textrm{T}}$, the distribution of IoT devices within  the intersection region still follows $\Phi_{\textrm{T}}$. Therefore, we have   
\begin{align} 
	&\mathcal{L}_{I_{\textrm{T}}}(s) \nonumber \\
	&= \mathbb{E}_{N_I,D_{T_n\textrm{-}S}} \! \!\left[ \!\prod_{T_n\in\Phi_{\textrm{C}} {\backslash} T_m} \!\!\!\int_{H}^{r_{\!\textrm{max}}}\!\!\! \left(\!1\!+\!\frac{t_0 D_{T_n\textrm{-}S} }{N_{T\textrm{-}S}r_n^{\alpha_1}}\!\right)^{\!\!\!-N_{T\textrm{-}S}}\!\!\!\!\!\! \frac{2r_n}{r_{\!\textrm{max}}^2\!-\!H^2} \mathrm{d} r_n \!\right] \nonumber
\end{align}
\begin{align}\label{L3-1} 
	&= \mathbb{E}_{N_I}\!\! \left[\! \prod_{T_n\in\Phi_{\textrm{C}} {\backslash} T_m} \!\!\!\int_{H}^{r_{\!\textrm{max}}}\!\!\!   \left( \left(\!1\!+\!\frac{t_0 G_{\rm t}G_{S} }{N_{T\textrm{-}S}r_n^{\alpha_1}}\!\right)^{\!\!\!-N_{T\textrm{-}S}} \!\!\!\frac{\varphi_{\textrm{t}}}{2\pi} \right.\right. \nonumber\\
	&\hspace*{3mm} + \left.\left. \left(\!1\!+\!\frac{t_0 g_{\rm t}G_{S} }{N_{T\textrm{-}S}r_n^{\alpha_1}}\!\right)^{\!\!\!-N_{T\textrm{-}S}}\!\!\! \Big(1-\frac{\varphi_{\textrm{t}}}{2\pi}\Big) \right) \frac{2r_n}{r_{\!\textrm{max}}^2\!-\!H^2} \mathrm{d} r_n \!\right] .
\end{align}

The point distribution of the interfering users can be described by a BPP. Therefore, the number of interfering users $N_I$ follows a binomial distribution with a certain probability of success. The probability of success can be expressed as
\begin{align}\label{L3-2} 
	P_I =& \frac{\mathcal{S}(\mathcal{A}_{\textrm{C}})}{\mathcal{S}(\mathcal{A}_{\textrm{T}})} = \frac{2\pi r_{\textrm{e}}^2(1-\cos(\theta_3))}{2\pi r_{\textrm{e}}^2(1-\cos(\theta_1))}=\frac{1-\cos(\theta_3)}{1-\cos(\theta_1)},
\end{align}
where $\cos(\theta_3)=\frac{H+r_{\textrm{e}}-r_{\textrm{max}} \cos\big(\frac{\varphi_{\textrm{s}}}{2}\big)}{r_{\textrm{e}}}$ and $r_{\textrm{max}}$ is given in (\ref{L0-4}).

Then $\mathcal{L}_{I_{\textrm{T}}}(s)$ can be derived as follows:
\begin{align} \label{L3-3}
	&\mathcal{L}_{I_{\textrm{T}}}(s) {=} \sum_{n_I=1}^{N_{\textrm{T}}-1}\! \binom{N_{\textrm{T}}-1}{n_I} (P_I)^{n_I}\!\! \left(\int_{H}^{r_{\textrm{max}}} \!\!\left(\!\! \left(\!1\!+\!\frac{t_0 G_{\rm t}G_{S} }{N_{T\textrm{-}S}r_n^{\alpha_1}}\!\right)^{\!\!\!-N_{T\textrm{-}S}} \!\!\! \right.\right. \nonumber\\
	&{\times}\left.\left. \frac{\varphi_{\textrm{t}}}{2\pi} +\left(\!1\!+\!\frac{t_0 g_{\rm t}G_{S} }{N_{T\textrm{-}S}r_n^{\alpha_1}}\!\right)^{\!\!\!-N_{T\textrm{-}S}} \!\!\!\! \left(1-\frac{\varphi_{\textrm{t}}}{2\pi}\right)\!\! \right) 
	\frac{2r_n}{r_{\textrm{max}}^2-H^2} \mathrm{d} r_n\!\!\right)^{n_I} \nonumber\\
	&{\times}(1-P_I)^{N_{\textrm{T}}-1-n_I} .	 
\end{align}

\begin{Theorem}\label{T1}	
	The CP of the serving satellite to any IoT device is given as 
	\begin{align}\label{eqT1} 
		P_{\rm{cov}}^{\rm{T}\textrm{-}\rm{S}} =& (1-P_0)\, \mathbb{P}( \mathrm{SINR}_1\geq T_1|R_{T_m}\leq r_{\rm{max}}) ,
	\end{align}
	where $T_1$ is the SINR threshold of the T-S link and
	\begin{align}\label{eqT1.1} 
		& \mathbb{P}(\mathrm{SINR}_1\geq T_1|R_{T_m}\leq r_{\rm{max}}) \nonumber \\
		& \hspace*{2mm} = \int_H^{r_{\rm{max}}} \sum_{n=1}^{N_{T\textrm{-}S}}(-1)^{n+1} \binom{N_{T\textrm{-}S}}{n}\mathcal{L}_{I_{\textrm{T}}}(s) \nonumber \\ 
		& \hspace*{5mm} \times\! N_{\rm{S}} \left(\! 1-\frac{ r_m^2-H^2}{4(r_{\rm{e}} +H)r_{\rm{e}}}\right)^{\!N_{\rm{S}}-1} \!\! \frac{ r_m }{2(r_{\rm{e}}+H)r_{\rm{e}}}\mathrm{d}r_m ,	
	\end{align}
	with $\mathcal{L}_{I_{\textrm{T}}}(s)$ being given by (\ref{L3-3}), and $s = \frac{n \eta T_1 16 \pi^2 f_1^2  r_m^{\alpha_1}}{p_{\rm{t}} D_{T_m\textrm{-}S}\,c^2}$. The complete expression of $P_{\rm{cov}}^{\rm{T}\textrm{-}\rm{S}}$ of the T-S link is given by (\ref{cov1}). 
\end{Theorem}

\begin{proof}
	See Appendix~\ref{ApC}.
\end{proof}

\subsection{CP over the S-ES Link}\label{S4.2}

\begin{Lemma}\label{L4}
	The Laplace transform of random variable $I_{\textrm{S}}$ is given by\setcounter{equation}{34}
	\begin{align} \label{eqL4} 
		\mathcal{L}_{I_{\rm{S}}}({s}') =& \sum_{n_I=1}^{N_{\rm{S}}-1} \binom{N_{\rm{S}}}{n_I}\left(\frac{1-\cos(\theta_3)}{2}\right)^{n_I} \nonumber \\
		&\times \int_{H}^{r_{\rm{max}}} \frac{(2\bar{c}q)^q(1+2\bar{c}{t_0}'  r_n'^{-\alpha_2})^{q-1}}{\big((2\bar{c}q+\Omega )(1+2\bar{c}{t_0}' r_n'^{-\alpha_2})-\Omega\big)^q} \nonumber \\
		&\times \frac{2r_n'}{r_{\rm{max}}^2\! -\! H^2}\mathrm{d}r_n' \left(\frac{1\! +\! \cos(\theta_3)}{2}\right)^{N_{\rm{S}}-n_I-1}\!\!\!,
	\end{align}
	where	$t_0'\! =\! \frac{t\,\zeta(\beta-\delta)\, T_2 p_n D_{S_n\textrm{-}ES}\,r_m'^{\alpha}}{p_{\rm{s}} D_{S_m\textrm{-}ES}}$ and $\cos(\theta_3)\! =\!\frac{H+r_{\rm{e}}-r_{\rm{max}} \cos\big(\frac{\varphi_{\rm{s}}}{2}\big)}{r_{\rm{e}}}$ .
\end{Lemma}

\begin{proof}
	See Appendix~\ref{ApF}. 
\end{proof}

\begin{Theorem}\label{T2}
	The CP for an arbitrarily located satellite capable of communicating with the ES under the SR fading channel is given by
	\begin{align}	\label{eqT2} 
		P_{\rm{cov}}^{\rm{S}\textrm{-}\rm{ES}} &{\triangleq} \mathbb{P}( \mathrm{SINR}_2\geq T_2) \nonumber \\
		&\hspace*{-1mm}= 1 -\! \sum_{k=0}^{\infty } \frac{\Psi\left(k\right)}{(\beta-\delta) ^{k+1}} \Gamma (k+1)\sum_{t=0}^{k+1}\binom{k+1}{t}   \nonumber\\
		&\hspace*{2mm}	\times (-1)^t \int_H^{r_{\rm{max}}} \!\! 	\mathcal{L}_{I_{\rm{S}}}({s}') \frac{2r_m'}{r_{\rm{max}}^2\! -\!H^2} \mathrm{d}r_m',	
	\end{align}
	where $T_2$ is the SINR threshold of the S-ES link, $\Psi (k)\! =\! \frac{\left(-1\right)^{k}\kappa\delta^{k}}{\left(k!\right)^{2}} {\rm Ps}(1\! -\! q)_{k}$, with $\kappa$ from (\ref{eqSSF}),  ${s}'\! =\! \frac{t\,\zeta(\beta-\delta)\, T_2}{p_{\rm{s}} D_{\rm{S}_m\textrm{-}\rm{ES}} \,l_2(r_m')}$, $\zeta\! =\! (\Gamma(k+2))^{-\frac{1}{k+1}}$, 	$\mathcal{L}_{I_{\rm{S}}}({s}')$ is given by (\ref{eqL4}), $I_{\textrm{S}}$ is the interference power of the interfering satellites visible to the ES,  $r_{\rm{max}}$ is given in (\ref{L0-4}), and the complete exression of $P_{\rm{cov}}^{\rm{S}\textrm{-}\rm{ES}}$ for the S-ES link is given by (\ref{eqCPs-gw}).	
\end{Theorem}

\begin{proof}
	See Appendix~\ref{ApE}.
\end{proof}

\begin{figure*}[!tp]\setcounter{equation}{36}
	\vspace*{-1mm}
	\begin{align}\label{eqCPs-gw} 
		&	P_{\mathrm{cov}}^{\textrm{S-ES}} =1 -\! \sum_{k=0}^{\infty }\,\, \frac{\left(-1\right)^{k}\kappa\delta^{k}(1-q)_{k}}{\left(k!\right)^{2}(\beta-\delta) ^{k+1}}\,\, \Gamma (k+1)\,\,\sum_{t=0}^{k+1}\,\,\binom{k+1}{t} (-1)^t \sum_{n_I=1}^{N_{\textrm{S}}-1} \binom{N_{\textrm{S}}}{n_I}\left(\frac{r_{\textrm{max}} \cos(\frac{\varphi_{\textrm{s}}}{2})-H}{2r_{\textrm{e}}}\right)^{n_I} \nonumber \\
		&\hspace*{3mm}\times\! \left(\! \frac{H\! +\! 2r_{\textrm{e}}\! -\! r_{\textrm{max}} \cos(\frac{\varphi_{\textrm{s}}}{2})}{2r_{\textrm{e}}}\! \right)^{\! N_{\textrm{S}}-n_I-1}\!\!\! \int_{H}^{r_{\textrm{max}}}\!\!\! \int_{H}^{r_{\textrm{max}}}\!\!\!\! \frac{(2\bar{c}q)^q\left(1\! +\! \frac{2\bar{c} t\,\zeta(\beta-\delta) T_2 p_n r_m'^{\alpha}}{p_{\textrm{s}} r_n'^{\alpha_2}} \right)^{q-1}}{\left((2\bar{c}q\! +\! \Omega )\left(1\! +\! \frac{2\bar{c} t \zeta(\beta-\delta)\, T_2 p_n r_m'^{\alpha}}{p_{\textrm{s}} r_n'^{\alpha_2}} \right)^{q-1}\! -\! \Omega\right)^q} \frac{4r_m'r_n'}{(r_{\textrm{max}}^2-H^2)^2} \mathrm{d}r_n'\mathrm{d}r_m'.
	\end{align} 		
	\hrulefill
	\vspace*{-4mm}
\end{figure*}

\subsection{E2E Coverage Probability}\label{S2.4}

In our model, the communication link is established from the ground IoT devices to the satellite, which then transparently forwards the signal to the ground station. Consequently, the E2E coverage probability is the product of the coverage probabilities for both the uplink from IoT devices to the satellite and the downlink from the satellite to the ES. Therefore, the E2E coverage probability can be expressed as:\setcounter{equation}{37}
\begin{align}\label{eqe2e} 
	P_{\textrm{cov}}^{\,{\textrm{E2E}}} =& \,P_{\textrm{cov}}^{\textrm{T-S}}\, P_{\textrm{cov}}^{\textrm{S-ES}} .
\end{align}

\begin{remark}\label{Rm2}
	The derived CP expression offers valuable insights for the overall system design. By modeling IoT devices within finite regions and considering satellite coverage limitations, we can more accurately assess how system parameters, such as the number of devices and  satellite coverage, affect the system's performance. This helps identify critical factors that influence communication reliability, guiding satellite constellation deployment and device distribution optimization. These insights enable a more efficient design of IoT-over-satellite systems, ensuring robust connectivity even in constrained environments. Additionally, this approach can be applied to other communication systems with similar spatial limitations, providing a flexible tool for network design.
\end{remark}

\section{Average Ergodic Rate}\label{S5}	

We now focus on the average achievable date rate. The AER, measured in bits/s/Hz, also known as the Shannon throughput, represents the mean data rate achievable over the system. It corresponds to the ergodic capacity for a fading communication link, normalized to a unit bandwidth. The AER is defined as follows\setcounter{equation}{38}
\begin{eqnarray}\label{eqAER}	
	\bar{C} \triangleq \mathbb{E}\left[ \log_2\left(1+\text{SINR}\right) \right] .
\end{eqnarray}

\subsection{AER over the T-S Link}\label{S5.1}

\begin{Theorem}\label{T3}
	The average rate of an arbitrary IoT device in the service link of a Nakagami fading channel is given by
	\begin{align}\label{eqT3} 
		\bar{C}^{\rm{T}\textrm{-}\rm{S}} =& (1-P_0) \, \mathbb{E}[ \log_2(1+\mathrm{SINR}_1)|R_{T_m}\leq r_{\rm{max}}] , 
	\end{align}
	in which $P_0$ is given in (\ref{P0}) and
	\begin{align}\label{eqT1.2} 
		& \mathbb{E}[ \log_2(1+\mathrm{SINR}_1)|R_{T_m}\leq r_{\rm{max}}] \nonumber \\
		& \hspace*{2mm}= \int_{H}^{r_{\rm{max}}}\!\!\! \int_{t>0} \sum_{n=1}^{N_{T\textrm{-}S}}(-1)^{n+1} \binom{N_{T\textrm{-}S}}{n}	\mathcal{L}_{I_{\rm{T}}}(s_1) \nonumber \\
		&\hspace*{5mm}\times\! \left(\! 1\! -\! \frac{r_m^2-H^2}{4(r_{\rm{e}}\! +\! H)r_{\rm{e}}}\right)^{N_{\rm{S}}-1}\!\!\! \frac{N_{\rm{S}} r_m }{2(r_{\rm{e}}\! +\! H)r_{\rm{e}} } \mathrm{d}t \,\mathrm{d}r_{m}, 
	\end{align}
	where $s_1\! =\! \frac{n \eta (2^t\! -\! 1) 16 \pi^2 f_1^2  r_m^{\alpha_1}}{p_{\rm{t}} D_{T_m\textrm{-}S}\, c^2} $, and $\mathcal{L}_{I_{\rm{T}}}(s_1)$ is obtained by  replacing $s$ in $\mathcal{L}_{I_{\rm{T}}}(s)$ with $s_1$.
\end{Theorem}

\begin{proof}
	See Appendix~\ref{ApG}.
\end{proof}

\subsection{AER over the S-ES Link} \label{S5.2}

\begin{Theorem}\label{T4}
	The AER of an arbitrary satellite node in the feeder link of an SR fading channel is given by
	\begin{align}\label{eqT3-m} 
		& \bar{C}^{\rm{S}\textrm{-}\rm{ES}}\!\! =\!\! \int_{t>0}\!\! \left(\!\!1\! -\!\! \int_H^{r_{\rm{max}}}\!\!\!\! \!\mathbb{E}\!\left[ \sum_{k=0}^{\infty} \frac{\Psi\left(k\right)}{(\beta\! -\! \delta) ^{\,k+1}} \Gamma (k\! +\! 1)\frac{2r_m'}{r_{\rm{max}}^2\! -\! H^2} \right.\right. \nonumber \\ 	 
		&\times\!\! \left. \left. \sum_{u=0}^{k+1}\!\! \binom{k+1}{u} (-1)^u \mathbb{E}_{I_{\rm{S}}}\!\! \left[\! \exp(-{s}_1'(I_{\rm{S}}\! +\!\sigma^2)) \right]\! \right]\!	\mathrm{d}r_m'\!\! \right)\! \mathrm{d}t ,\!
	\end{align}
	where 
	\begin{align}\label{eqAdd} 
		\mathbb{E}_{I_{\rm{S}}}\! \! \left[ \exp\!\left(-s_1' {(I_{\rm{S}}\!+\!\sigma^2\!)}\!\right)\!\right]\!\! &=\!\mathbb{E}_{I_{\rm{S}}}\! \! \left[\exp\left(-s_1'\! {I_{\rm{S}}}\right)\right]\!\exp\left(-\!s_1'\! {\sigma^2}\right)\nonumber\\
		&=\mathcal{L}_{I_{\rm{S}}}({s_1}')\!\exp\left(-\!s_1'\! {\sigma^2}\right),
	\end{align} 
	with ${s}_1'\! =\!\frac{u\,\zeta(\beta-\delta)\, (2^t\! -\! 1)}{p_{\rm{s}} D_{S_m\textrm{-}ES} \,l_2(r_m')}$, and $\mathcal{L}_{I_{\rm{S}}}({s_1}')$ is obtained by replacing $s'$ in $\mathcal{L}_{I_{\rm{S}}}({s}')$ with $s_1'$.
\end{Theorem}

\begin{proof}
	See Appendix~\ref{ApH}.
\end{proof}	

\begin{remark}\label{Rm3}
	It is worth noting that when calculating the AER for the S-ES link, we must consider noise. This is because satellites are distributed according to a BPP model on a spherical surface at an altitude $H$ from the ground. In the field of view of the ES, there may be no interfering satellites present. In such cases, the presence of noise $\sigma^2$ is essential to give the AER calculation meaningful significance.
\end{remark}

\section{Numerical Results}\label{S6}

We validate the derived theoretical expressions using Monte Carlo simulations with 50,000 iterations. The results from the analytical derivations in Sections~\ref{S4} and \ref{S5} are labeled as `Analysis', while the Monte Carlo results are labeled as `Simulation'. The default system parameters used in the simulations are listed in Table~\ref{Table2}, unless otherwise specified. In all the experiments, the Monte Carlo simulation results closely match to the corresponding analytical results, which demonstrates the accuracy of our theoretical analysis.

\begin{table}[h]\scriptsize
	\vspace*{-2mm}
	\caption{Default Simulation System Parameters.} 
	\label{Table2} 
	\vspace*{-6mm}
	\begin{center}
		\begin{tabular} {c|c|c}
			\Xhline{1.2pt}
			\toprule
			\textbf{Notation} & \textbf{Parameter} & \textbf{Default Value} \\ \midrule
			$H$                        & Height of satellites                      & 400\,km \\
			$R_{\rm{c}}$               & Radius of $\mathcal{A}_{\textrm{T}}$      & 200\,km \\
			$N_{\textrm{T}}$           & Number of IoT devices                     & 5000 \\
			$N_{\textrm{S}}$           & Number of satellites                      & 3000 \\
			$\varphi_{\textrm{s}}$     & Antenna beamwidth of satellites            & $25^{\circ}$ \\
			$N_{T\textrm{-}S}$          & Nakagami fading parameter                 & 2 \\
			$\mathcal{SR}(\bar{c},q,\Omega)$ & SR fading parameter                       & $\mathcal{SR}(0.158,1,0.1)$ \\
			$\sigma^2$                   & AWGN's power spectral density             & -98\,dBm  \\
			$\alpha_1$, $\alpha_2$     \!\!\!\!&\!\!\!\! Path-loss exponents of T-S and S-ES links \!\!\!\!&\!\!\!\! 2, 2 \\
			$T_{1}$, $T_{2}$           & SINR threshold of T-S and S-ES links      & variable \\ \bottomrule
		\end{tabular}	
	\end{center}
	\vspace*{-1mm}
\end{table}	

\begin{figure}[!t]\setcounter{figure}{2}
	\vspace*{-2mm}
	\begin{center}
		\includegraphics[width=1\columnwidth]{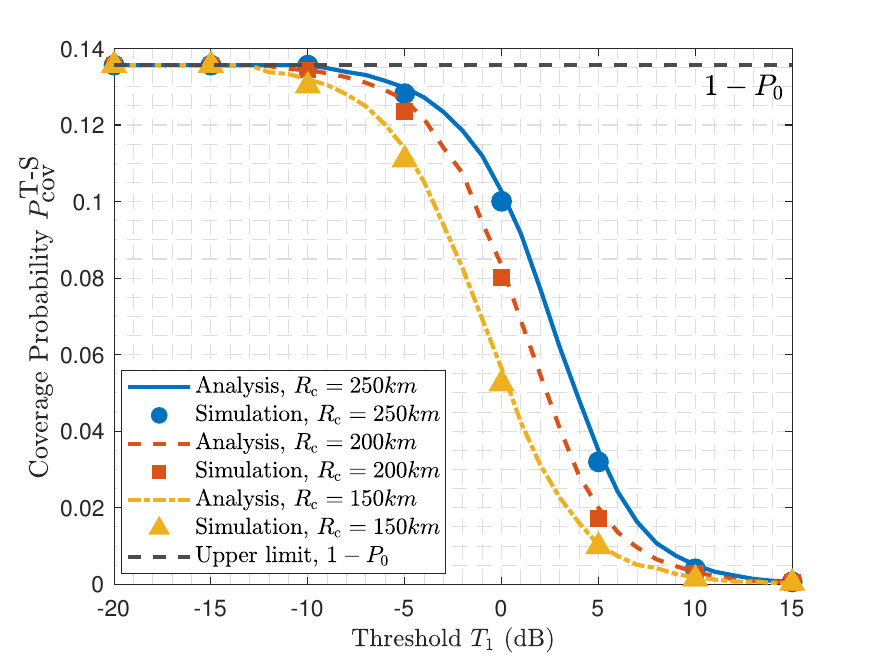}
	\end{center}
	\vspace*{-6mm}
	\caption{CP as function of threshold $T_1$, given different finite-size area radii $R_{\rm{c}}$.}
	\label{fig:3} 
\end{figure}

\subsection{Performance of the T-S Link} \label{S6.1}

Fig.~\ref{fig:3} illustrates the CP as the function of the SINR threshold $T_{1}$\,{\footnote{The choice of SINR threshold is influenced by factors such as the link budget, signal propagation characteristics and the specific environment of the IoT-satellite communication system. In practice, this threshold is chosen based on system parameters and specific deployment scenarios. In our analysis, we consider the SINR threshold ranging from -20\,dB to\,15 dB to explore the system's behavior across different link budget conditions.}} given three different values of the radius $R_{\rm{c}}$ of the finite-size area $\mathcal{A}_{\textrm{T}}$. As expected, increasing the SINR threshold $T_{1}$ leads to a decrease in the CP. This decrease is due to the inverse relationship between $T_{1}$ and the probability of achieving an SINR exceeding the specified threshold value. It can also be seen that an increase in the coverage radius of $\mathcal{A}_{\textrm{T}}$ leads to a noticeable enhancement in the CP. This is attributed to the fact that when the total number of IoT devices remains constant, a larger finite-size area results in a reduced density of IoT devices. Consequently, there are fewer interfering IoT transmitters simultaneously attempting to access the serving satellite, leading to a decrease in MUI and thus an increase in the CP. Note that the upper limit of the CP is not 1 but $1\! -\! P_0$. This is because the satellite's coverage is limited, and when $R_{T_m}\! >\! r_{\textrm{max}}$, the nearest satellite cannot communicate with the target IoT device,  resulting in the zero CP. Moreover, as shown in (\ref{P0}), $P_0$ is independent of $R_{\rm{c}}$, and the upper limits of CP corresponding to different $R_{\rm{c}}$ values are identical.
\begin{figure}[!t]
	\begin{center}
		\includegraphics[width=1\columnwidth]{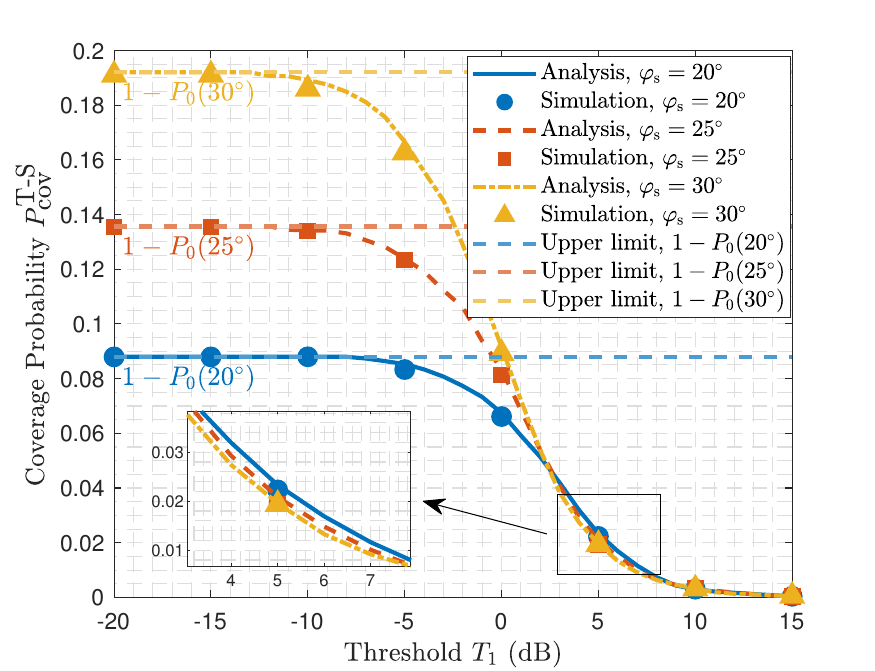}
	\end{center}
	\vspace*{-6mm}
	\caption{CP as function of threshold $T_1$, given different satellite antenna beamwidth $\varphi_{\textrm{s}}$.}
	\label{fig:4}
	\vspace*{-1mm}
\end{figure}
\begin{figure}[!t]
	\vspace*{-2mm}
	\begin{center}
		\includegraphics[width=1\columnwidth]{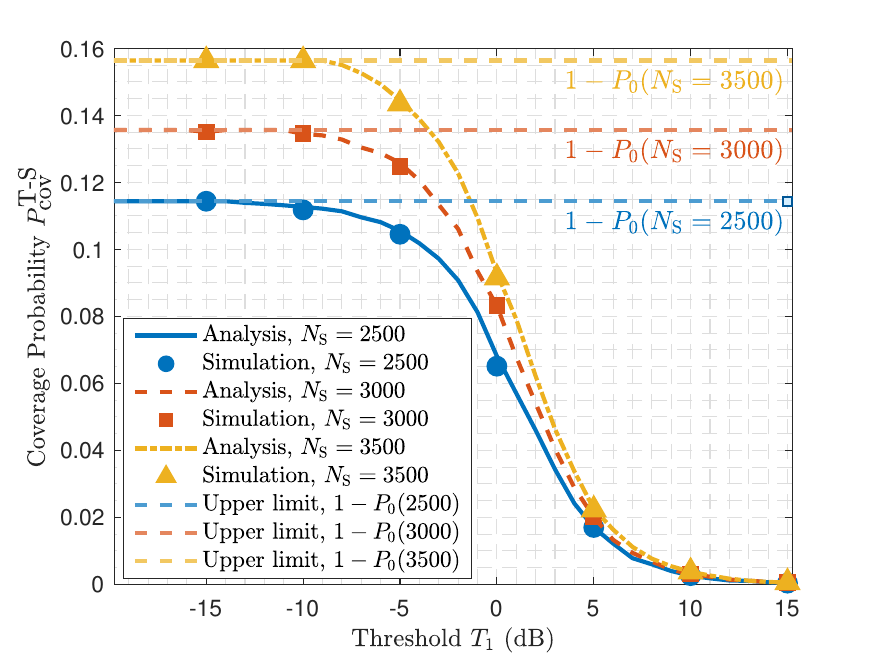}
	\end{center}
	\vspace*{-6mm}
	\caption{CP as function of threshold $T_1$, given different numbers of satellites $N_{\textrm{S}}$.}
	\label{fig:5}
	\vspace*{-4mm}
\end{figure}

\begin{figure}[!t]
	\begin{center}
		\includegraphics[width=1\columnwidth]{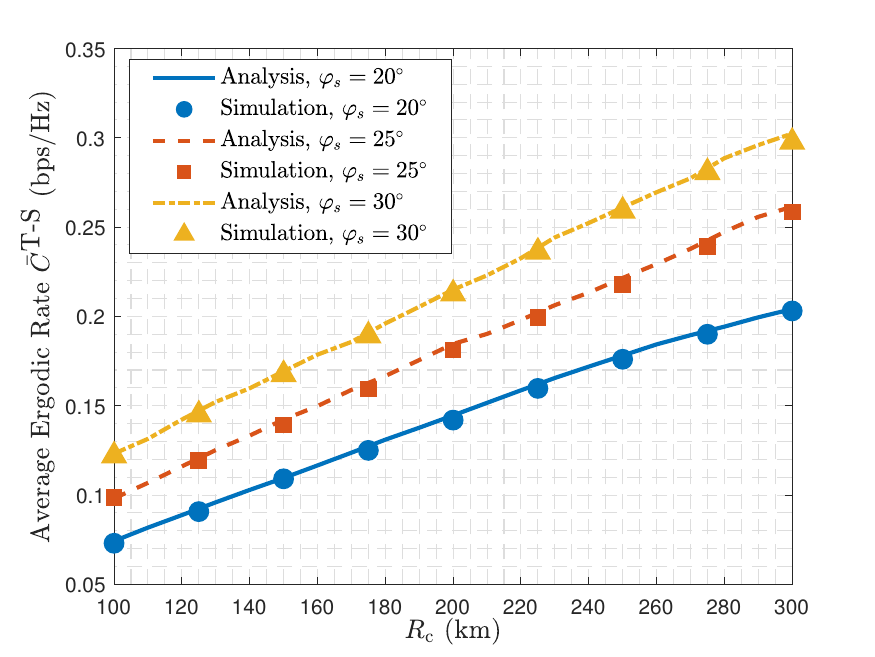}
	\end{center}
	\vspace*{-7mm}
	\caption{AER as function of finite-size area's radius $R_{\rm{c}}$, given different antenna beamwidth of satellites $\varphi_{\textrm{s}}$.}
	\label{fig:6}
	\vspace*{-1mm}
\end{figure}

Fig.~\ref{fig:4} plots the CP as the function of the SINR threshold $T_{1}$ for three different values of the satellite antenna beamwidth $\varphi_{\textrm{s}}$. The results suggest that increasing the satellite coverage reduces the T-S link's CP. This is attributed to the impact of $\varphi_{\textrm{s}}$ on the SINR, where a larger coverage service satellite leads to increased MUI and consequently reduced CP. Also, due to the fact that $P_0$ is determined by $\varphi_{\textrm{s}}$, the upper limit values of the three CP curves are different, which can be explained in formulas (\ref{P0}) and (\ref{eqT1}). This is because as $\varphi_{\textrm{s}}$ increases, the communication distance of the satellite increases, giving IoT devices a higher probability of connecting to the serving satellite. However, the increase in communication range also means an increase in MUI. Therefore, the larger the $\varphi_{\textrm{s}}$, the faster the CP reduces as $T_1$ increases.

Additionally, Fig.~\ref{fig:5} investigates the influence of the number of satellites $N_{\textrm{S}}$ on the achievable CP. As expected, increasing $N_{\textrm{S}}$ results in an increase in the CP because a higher number of satellites increases the likelihood that the nearest satellite to the target IoT device will be within communication range, which also increases the upper limit $1 - P_0$ of CP. 

Fig.~\ref{fig:6} depicts the AER as the function of the finite-size area's radius $R_{\rm{c}}$, given three different values for satellite antenna beamwidth $\varphi_{\textrm{s}}$. As expected, increasing $R_{\rm{c}}$ increases $\bar{C}^{\textrm{T-S}}$. The impact of $\varphi_{\textrm{s}}$ on the achievable AER is clearly shown in Fig.~\ref{fig:6}. Evidently, increasing the available area $\mathcal{A}_{\textrm{T}}$ results in an increase in the AER, because given the fixed number of IoT devices $N_{\textrm{T}}$, a larger distribution area results in a lower density and a decrease in the number of interfering devices. Also the larger the $\varphi_{\textrm{s}}$, the greater the communication distance. This increases the probability that the satellite can establish a connection, thereby leading to a higher AER.

\begin{figure}[!t]
	\vspace*{-1mm}
	\begin{center}
		\includegraphics[width=1\columnwidth]{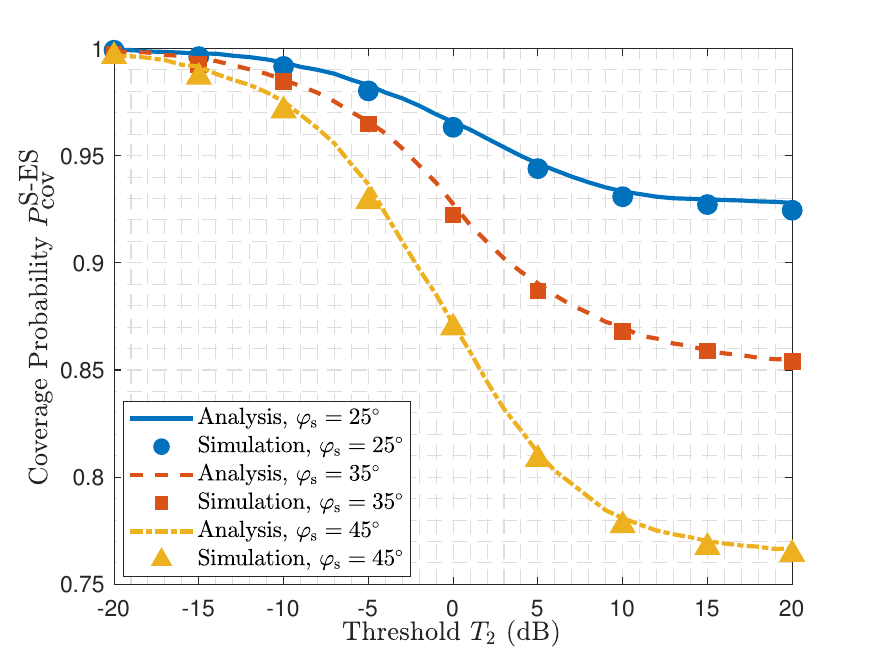}
	\end{center}
	\vspace*{-6mm}
	\caption{CP as function of threshold $T_2$, given different satellite antenna beamwidth $\varphi_{\textrm{s}}$.}
	\label{fig:7}
	\vspace*{-3mm}
\end{figure}

\begin{figure}[!t]
	\begin{center}
		\includegraphics[width=1\columnwidth]{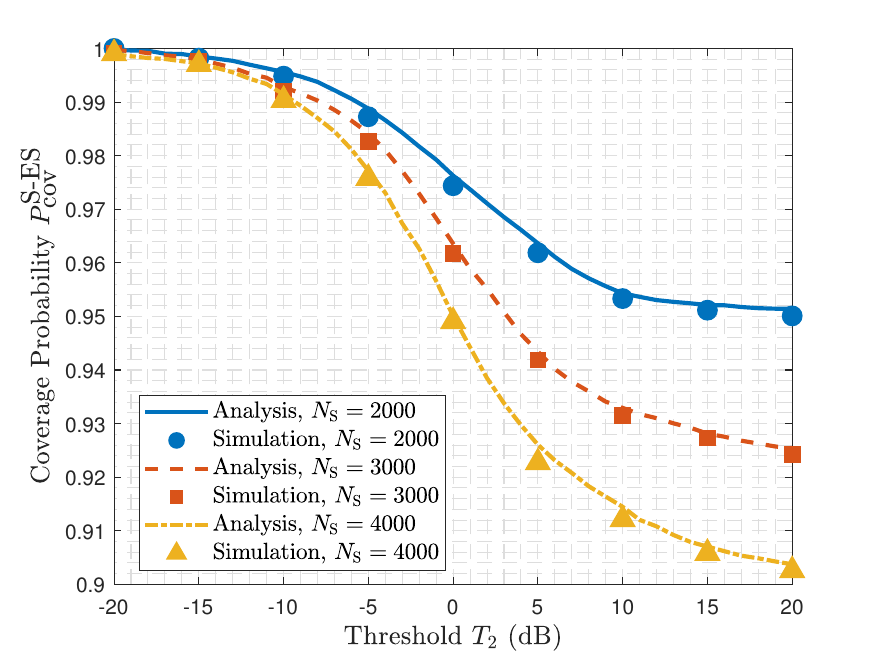}
	\end{center}
	\vspace*{-6mm}
	\caption{CP as function of threshold $T_2$, given different numbers of satellites $N_{\textrm{S}}$.}
	\label{fig:8}
	\vspace*{-3mm}
\end{figure}

\subsection{Performance of the S-ES Link} \label{S6.2}

Fig.~\ref{fig:7} depicts the CP as the function of the SINR threshold $T_{2}$ for three different values of the satellite overage range $\varphi_{\textrm{s}}$. As anticipated, increasing $T_2$ results in a reduction in CP.  Besides, an increase in $\varphi_{\textrm{s}}$ leads to a noticeable reduction in the CP. This is because an increase in the satellite beam coverage expands the ES's visible range, resulting in a higher number of interfering satellites within this range. This leads to higher MUI, causing a decrease in the CP. It is worth noting that the CP curve does not reach zero, rather it reaches a lower limit value, when $T_{2}$ increases to very large value. This is due to the limited number of satellites and their restricted coverage areas. Given that there must always be a satellite providing ES service, it is highly likely that no interfering satellites exist within the ES's visible range. As the visible range increases, this limit value of CP will decrease.

Fig.~\ref{fig:8} plots the CP as the function of $T_{2}$ for three different numbers of satellites $N_{\textrm{S}}$. As expected, increasing $N_{\textrm{S}}$ reduces the achievable CP since a higher number of satellites increases the number of interfering satellites. This reduces the SINR of the S-ES link, causing a reduction in the CP.

\begin{figure}[!b]
	\vspace*{-4mm}
	\begin{center}
		\includegraphics[width=1\columnwidth]{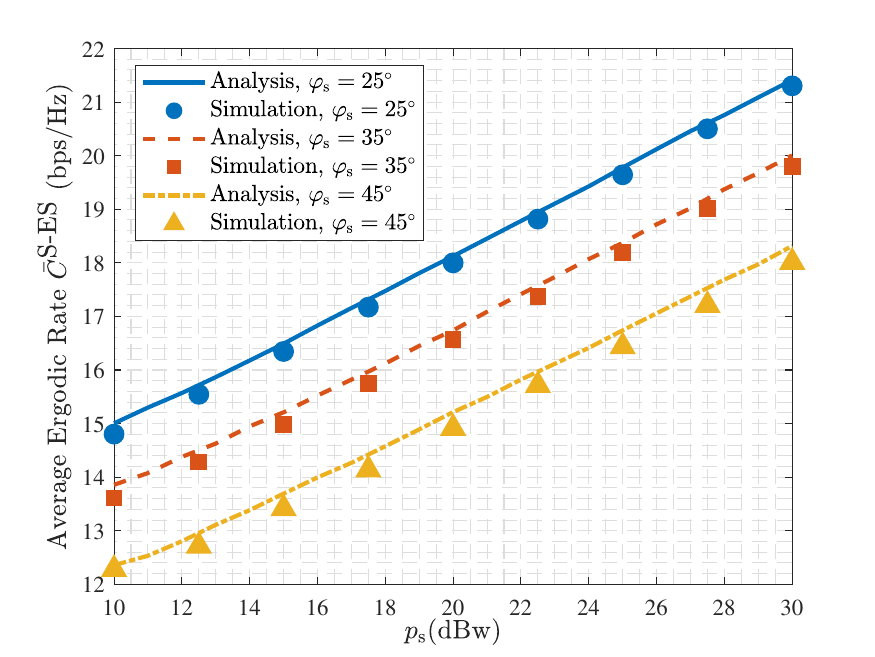}
	\end{center}
	\vspace*{-6mm}
	\caption{AER as function of power of target satellite $p_{\textrm{s}}$, given different satellite antenna beamwidth $\varphi_{\textrm{s}}$.}
	\label{fig:9} 
	\vspace*{-1mm}
\end{figure}

Fig.~\ref{fig:9} shows the AER as the function of the target satellite power $p_{\textrm{s}}$, given three different values for the satellite antenna beamwidth $\varphi_{\textrm{s}}$. As expected, increasing $p_{\textrm{s}}$ increases $\bar{C}^{\textrm{S-ES}}$. In addition, increasing $\varphi_{\textrm{s}}$ leads to an increase in the number of interfering satellites, which decreases the AER. 

\begin{figure}[!b]
	\vspace*{-6mm}
	\begin{center}
		\includegraphics[width=1\columnwidth]{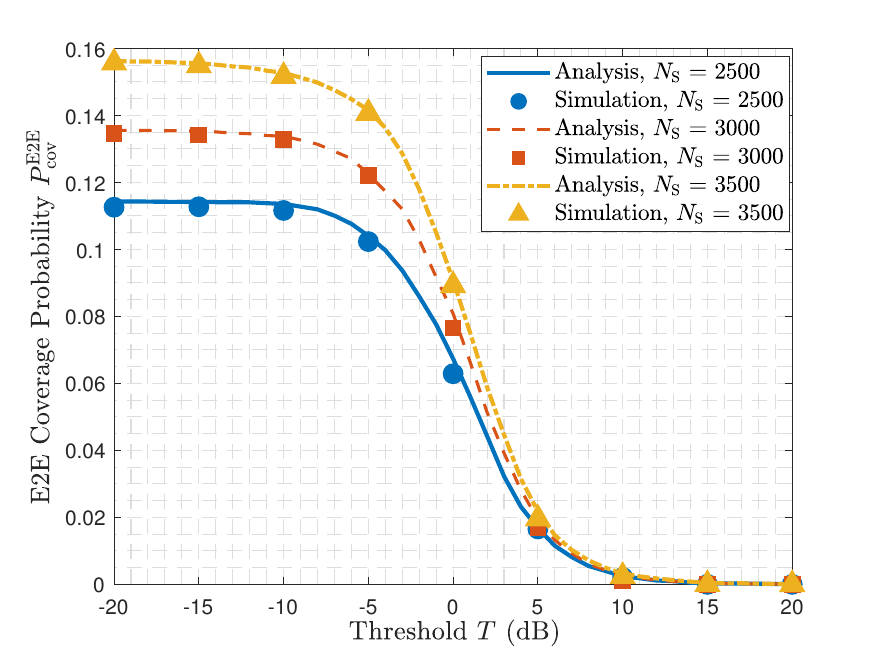}
	\end{center}
	\vspace*{-6mm}
	\caption{E2E coverage probability as function of threshold $T$, given different numbers of satellites  $N_{\textrm{S}}$.}
	\label{fig:10}
	\vspace*{-1mm}
\end{figure}

\subsection{E2E Performance} \label{S6.3}

In this set of simulations, we assume the threshold $T_1$ equals the threshold $T_2$, i.e., $T_1=T_2=T$. {\color{black}An identical threshold value can measure the overall communication quality of the system in a more intuitive manner, because it provides a consistent reference standard, representing the performance of both uplink and downlink based on the same threshold $T$. This consistency allows the E2E coverage probability to directly reflect the system's ability to meet overall service quality requirements. By avoiding the complex interactions introduced by different thresholds, it simplifies performance analysis and emphasizes the system's overall communication quality, making the results clearer and easier to interpret.}

Fig.~\ref{fig:10} plots the E2E coverage probability as the function of the SINR threshold $T$ for three different values of the number of satellites $N_{\textrm{S}}$. It can be seen that increasing $N_{\textrm{S}}$ results in an increase in the CP. The CP performance of the service and feeder links jointly influence the E2E coverage probability. As shown in Figs.~\ref{fig:5} and \ref{fig:8}, the impact of $N_{\textrm{S}}$ on CP is opposite in the service and feeder links. Specifically, increasing $N_{\textrm{S}}$ increases the CP of the service link, while increasing $N_{\textrm{S}}$ decreases the CP of the feeder link. Therefore, it is evident that the service link has a more significant impact on the overall system performance. This phenomenon can be explained as follows. When an IoT device selects the nearest satellite as its serving satellite, there is a certain probability ($P_0$) that no satellite exists within the device's visible range due to the limited coverage of satellites. In such cases, the IoT device cannot establish connection because no serving satellite is available. This probability $P_0$, specified by \eqref{P0}, quantitatively captures this limitation. However, once the signal is relayed via satellite to the ES, it is ensured that at least one satellite can communicate with the ES due to the design of satellite relay links and the broader coverage of the ES. Clearly, the availability of a satellite within the IoT device's range plays a critical role in determining the system's overall connectivity.

\begin{figure}[!t]
	\vspace*{-1mm}
	\begin{center}
		\includegraphics[width=1\columnwidth]{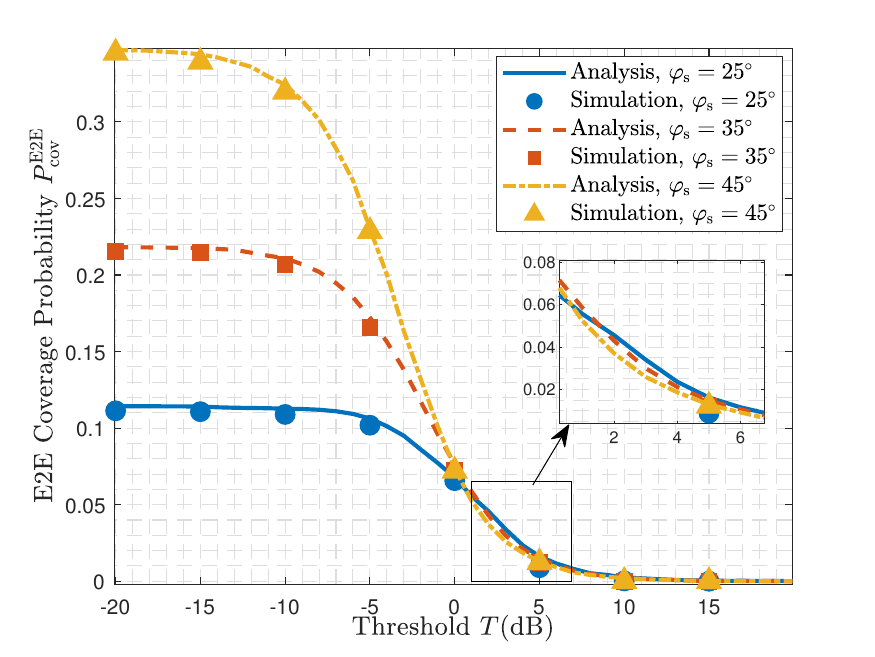}
	\end{center}
	\vspace*{-6mm}
	\caption{E2E coverage probability as function of threshold $T$, given different satellite antenna beamwidth $\varphi_{\textrm{s}}$.}
	\label{fig:11}
	\vspace*{-3mm}
\end{figure}

Fig.~\ref{fig:11} shows the E2E coverage probability as the function of the SINR threshold $T$, given three different satellite antenna beamwidth $\varphi_{\textrm{s}}$.   As expected, increasing $\varphi_{\textrm{s}}$ increases the E2E coverage probability since a larger satellite  coverage provides greater benefits to the service link, thereby leading to an overall increase in the E2E coverage probability. It can be observed from the coverage probability curves of the service link (see Fig.~\ref{fig:4}) and the feeder link (see Fig.~\ref{fig:7}) that the impact of satellite antenna beamwidth on the coverage probability of the service link is opposite to that of the feeder link. Furthermore, the downward trend shown in Fig. \ref{fig:11} is more similar to that of the service link shown in Fig. \ref{fig:4}, which means that the service link has a higher impact on the E2E performance than the feeder link. Meanwhile, due to the opposite trend observed regarding the feeder link, the decline rate of each curve in Fig.~\ref{fig:11} is significantly faster than that in Fig.~\ref{fig:4}. This implies that when jointly considering the service link and the feeder link, the dynamic change of the system performance may become more complex than that of only considering the service link or the feeder link. In other words, the opposite impact of satellite antenna beamwidth on the coverage probability of the feeder link complicates the overall coverage performance profile. The faster decline rate of the curves in Fig.~\ref{fig:11} indicates that although the overall trend is similar to that of the service link, the characteristics of the feeder link significantly weaken the system's coverage ability under a large satellite coverage. Similar insights can also be obtained by comparing the results of Fig.~\ref{fig:5} and  Fig.~\ref{fig:10}. 
	
	The above discussions suggest that in practical system design and resource allocation, a judicious and balanced parameter configuration on the two links is required. For example, when planning the satellite coverage, one cannot simply rely on the characteristics of the service link but should fully consider the reverse effect of the feeder link. Otherwise, as the coverage expands, the overall coverage probability may deteriorate at a faster rate. It also means that when optimizing system performance, the optimization strategies for the feeder link may be different from those for the service link. Customized methods need to be explored to enhance the stability of the feeder link under large coverage, thereby slowing down the decline rate of the overall coverage probability and ensuring reliable operation of the system under different coverage conditions.

\section{Conclusions} \label{S7}

In this paper, we have proposed a tractable approach for analyzing the performance of IoT-over-LEO satellite systems, where IoT devices are located in a finite-size ground area. We have analyzed and discussed the performance of the service link and the feeder link separately, as well as that of the E2E overall system. The following conclusions have been drawn from our theoretical analysis and simulations results.
\begin{enumerate}
	\item When the satellite coverage and the number of satellites are both limited, the coverage probability of the service link has a clear upper limit, which is determined by the probability $P_0$ that no satellite exists within the visible range of an IoT device. 
	
	\item An increase in the number of satellites and their coverage both lead to an improvement in coverage probability. However, when IoT devices simultaneously operate on the same frequency resource, a larger coverage also means more interfering devices, resulting in a faster decline of coverage probability. 
	
	\item The number of satellites and the satellite coverage have opposing effects on the coverage probabilities and average ergodic rates of the service and feeder links. While increasing the number of satellites or satellite coverage improves the service link's performance, it can exacerbate interference in the feeder link, leading to a decrease in its performance. 
	
	\item The service link has a more significant impact on the overall system performance because the limited satellite coverage results in a certain probability $P_0$ that no satellite is within an IoT device's visible range, preventing connectivity. In contrast, the link from satellite to ground station ensures communication due to broader coverage. Therefore, the availability of a serving satellite is crucial for system connectivity. 
\end{enumerate}

For practical IoT-over-LEO satellite system design, both the first and second insights imply that network planners need to carefully consider the trade-off between the cost of deploying more satellites and optimizing the achievable coverage of each satellite. 
For instance, in urban areas with a high density of IoT devices, a relatively smaller but well-planned satellite coverage can be adopted to balance coverage and interference, thus enhancing the overall system performance. The third insight is also crucial for the deployment of IoT-over-LEO satellite systems. It means operators need to coordinate the design of the service and feeder links. For example, in a system with a large number of IoT devices in a specific area, additional resources may be allocated to the feeder link, such as using more advanced interference-mitigation techniques, to counterbalance the negative impact on the feeder link when enhancing the service link. The fourth insight implies that efforts should be focused on improving the reliability of the service link. This could involve developing more accurate satellite-tracking algorithms for IoT devices to increase the probability of being connected to a satellite, or deploying relay nodes in areas with poor satellite visibility to bridge the communication gap.

In summary, our study offers theoretical guidance and valuable insights for IoT-over-LEO satellite system planning, deployment, and optimization in practice. 

\appendix

\subsection{Proof of Lemma~\ref{L1}}\label{ApA}

\begin{proof}
	Let $\mathcal{A}(x)$ represent the spherical cap with the radius $r_{\textrm{e}}\! +\! H$ such that the distance between any point on $\mathcal{A}(x)$ and the IoT device is equal to $x\! \in\! [H, \, 2r_{\rm{e}}\! + \!H]$, $a$ and $b$ signify the radius and height of $\mathcal{A}(x)$, respectively, and the surface area of $\mathcal{A}(x)$ be $\mathcal{S}(\mathcal{A}(x))$. If $x\! =\! H$, the surface area vanishes, i.e., $\mathcal{S}(\mathcal{A}(H))\! =\! 0$, while when $x\! =\! 2 r_{\textrm{e}}\! +\! H$, the surface area becomes a whole sphere with the radius $r_{\textrm{e}}\! +\! H$, which is the region where the satellites are located, i.e., $\mathcal{S}(\mathcal{A}( 2r_{\rm{e}}\! +\! H))\! =\! 4\pi (r_{\textrm{e}}\! +\! H)^2$. From basic geometry, we obtain 
	\begin{align} 
		\mathcal{S}(\mathcal{A}(x)) =& 2 \pi (r_{\textrm{e}}+H) b,  \label{A1.1} \\
		\mathcal{S}(\mathcal{A}(x)) =& \pi (a^2+b^2), \label{A1.2} \\
		x^2 =& (H-b)^2+a^2. \label{A1.3}
	\end{align}
	Substituting (\ref{A1.1}) and (\ref{A1.2}) into (\ref{A1.3}) leads to the relationship expression between $x^2$ and $\mathcal{S}(\mathcal{A}(x))$ as 
	\begin{align}\label{A1.4} 
		x^2 =& \frac{\mathcal{S}(\mathcal{A}(x))}{\pi}-2H\frac{\mathcal{S}(\mathcal{A}(x))}{2\pi (r_{\textrm{e}}+H)}+H^2  \nonumber \\
		=& \mathcal{S}(\mathcal{A}(x)) \frac{r_{\textrm{e}}}{\pi (r_{\textrm{e}}+H)}+H^2.
	\end{align}
	Thus, we have
	\begin{align}\label{A1.5} 
		\mathcal{S}(\mathcal{A}(x)) =& \frac{\pi (r_{\textrm{e}}+H)(x^2 -H^2)}{r_{\textrm{e}}}.
	\end{align}
	
	The probability that the distance between satellite $S$ and IoT device $T$ is less than $r_\textrm{s}$ is equivalent to the probability of the satellite $S$ being within $\mathcal{A}(r_\textrm{s})$, denoted as the success probability for $\mathcal{A}(r_\textrm{s})$, with $r_\textrm{s}\! \in\! [H, \, 2r_{\rm{e}}\! + \!H]$. This is defined as the ratio of the surface area of $\mathcal{S}(\mathcal{A}(r_\textrm{s}))$ to $\mathcal{S}(\mathcal{A}(2r_{\rm{e}}\! +\! H))$:
	\begin{align}\label{A1.6} 
		\mathbb{P}(R_{\textrm{S}} \leq r_\textrm{s}) = \frac{\mathcal{S}(\mathcal{A}(r_\textrm{s}))}{\mathcal{S}(\mathcal{A}(2r_{\rm{e}} + H))}= \frac{r_\textrm{s}^2-H^2}{4(r_{\textrm{e}}+H)r_{\textrm{e}}},
	\end{align}
	for $H\le r_\textrm{s}\! \le\! 2r_{\rm{e}}\! +\! H$. Clearly $\mathbb{P}(R_{\textrm{S}}\! \leq\! r_\textrm{s})\! =\! 0$ for $r_\textrm{s}\! <\! H$, and $\mathbb{P}(R_{\textrm{S}}\! \leq\! r_\textrm{s})\! =\! 1$ for $r_\textrm{s}\! >\! 2r_{\rm{e}}\! +\! H$. This leads to (\ref{L1-1}).
	
	Differentiating (\ref{L1-1}) with respect to $r_\textrm{s}$ leads to (\ref{L1-2}).
\end{proof}

\subsection{Proof of Lemma~\ref{L2}} \label{ApB}

\begin{proof}
	Let $\mathcal{A}'(x')$ denote the spherical cap with the radius $r_{\textrm{e}}$ such that the distance between any point on $\mathcal{A}'(x')$ and the satellite is equal to $x'\! \in\! [H,\, 2r_{\rm{e}}\! +\! H]$, $a_1$ and $b_1$ signify the radius and height of $\mathcal{A}'(x')$, and the surface area of $\mathcal{A}'(x')$ be $\mathcal{S}(\mathcal{A}'(x'))$. If $x'\! =\! H$, $\mathcal{S}(\mathcal{A}'(H))\! =\! 0$. When $x'\! =\! 2 r_{\textrm{e}}\! +\! H$, the surface area becomes the surface of the Earth, i.e., $\mathcal{S}(\mathcal{A}'( 2r_{\rm{e}}\! +\! H))\! =\! 4\pi r_{\textrm{e}}^2$. From basic geometry, we can obtain 
	\begin{align} 
		\mathcal{S}(\mathcal{A}'(x')) =& 2 \pi r_{\textrm{e}} b_1, \label{A2.1} \\
		\mathcal{S}(\mathcal{A}'(x')) =& \pi (a_1^2+b_1^2), \label{A2.2} \\
		{x'}^2 =& (H+b_1)^2+a_1^2. \label{A2.3}
	\end{align}
	From (\ref{A2.1})-(\ref{A2.3}), we obtain the relationship: 
	\begin{align}\label{A2.4} 
		\mathcal{S}(\mathcal{A}'({x'})) =& \frac{\pi r_{\textrm{e}}({x'}^2 -H^2)}{(r_{\textrm{e}}+H)}.
	\end{align}
	Similar to Appendix~\ref{ApA}, it is straightforward to derive (\ref{L2-1}).
\end{proof}

\subsection{Proof of Lemma~\ref{L3}}\label{ApD}

\begin{proof}
	$\mathcal{L}_{I_{\textrm{T}}}\left(s\right)$ can be written as
	\begin{align*} 
		& \mathcal{L}_{I_{\textrm{T}}}\!\left(s\right)\!= \! \mathbb{E}\!\!\left[ \prod_{T_n\in\Phi_{\textrm{C}} {\backslash} T_m}\!\!\!\!\! \exp\! \left( -s p_{\textrm{t}} D_{T_n\textrm{-}S}\left|h_{T_n\textrm{-}S }\right|^2 \overline{D}\,l_1(R_{T_n}) \right)\! \right] 
	\end{align*}
	\begin{align}\label{eqApD1} 
		&\overset{\mathrm{(a)}}{=}\! \mathbb{E}\!\! \left[ \prod_{T_n\in\Phi_{\textrm{C}} {\backslash} T_m}\!\!\!\!\!\mathbb{E}_{\left|h_{T_n\textrm{-}S }\right|^2}\!\!\left[\exp\! \left(\left|h_{T_n\textrm{-}S }\right|^2\!\!\left(-t_0 D_{T_n\textrm{-}S}  R_{T_n}^{-\alpha_1} \right) \right)\! \right]\! \right] \nonumber \\
		&\overset{\mathrm{(b)}}{=}\!\mathbb{E}_{N_I,R_{T_n},D_{T_n\textrm{-}S}}\!\! \left[ \!\prod_{T_n\in\Phi_{\textrm{C}} {\backslash} T_m}    \left(\!1\!+\!\frac{t_0 D_{T_n\textrm{-}S} }{N_{T\textrm{-}S}R_{T_n}^{\alpha_1}}\!\right)^{\!\!\!-N_{T\textrm{-}S}} \!\right],
	\end{align}
	where (a) is obtained by denoting $t_0\! =\! \frac{n\eta T_1r_m^{\alpha_1}\overline{D}}{D_{T_m\textrm{-}S}}$, and (b) is obtained by using the moment-generating function (MGF) of the normalized Gamma random	variable. 
\end{proof}

\subsection{Proof of Theorem~\ref{T1}}\label{ApC}

\begin{proof}
	
	From Bayes rule, we have
	\begin{align}\label{eqApC1} 
		P_{\rm{cov}}^{\rm{T}\textrm{-}\rm{S}} =& \mathbb{P}(R_{T_m} > r_{\rm{max}}) \mathbb{P}( \mathrm{SINR}_1\geq T_1|R_{T_m}> r_{\rm{max}}) \nonumber \\
		& + \mathbb{P}(R_{T_m} \le r_{\rm{max}}) \mathbb{P}( \mathrm{SINR}_1\geq T_1|R_{T_m}\leq r_{\rm{max}})  \nonumber \\
		=& (1-P_0)\, \mathbb{P}( \mathrm{SINR}_1\geq T_1|R_{T_m}\leq r_{\rm{max}}).
	\end{align}
	Note that
	\begin{align}\label{eqApC2} 
		&	\mathbb{P}(\mathrm{SINR}_1\geq T_1|R_{T_m}\leq r_{\textrm{max}}) \nonumber \\
		&\hspace*{2mm}= \mathbb{E}_{R_{T_m}}\left[ \mathbb{P}\left(\mathrm{SINR}_1 \geq  T_1|R_{T_m}=r_m, R_{T_m}\leq r_{\textrm{max}}\right) \right] \nonumber \\
		&\hspace*{2mm}=\! \int_H^{r_{\textrm{max}}}\!\! \mathbb{P}\left(\mathrm{SINR}_1 \geq  T_1|R_{T_m}=r_m\right) f_{R_{T_m}}(r_m) \mathrm{d}r_m .
	\end{align}
	As $f_{R_{T_m}}(r_m)$ is given in (\ref{L1-4}), we only need to derive:
	\begin{align}\label{eqApC3} 
		& \text{$\mathbb{P}\left(\mathrm{SINR}_1 \geq T_1|R_{T_m}=r_m\right)$} \nonumber \\
		&\hspace*{2mm}= 1 - \mathbb{P}\!\left(|h_{T_m\textrm{-}S}|^2 \le\frac{T_1{I_{\textrm{T}}}}{p_{\textrm{t}} D_{T_m\textrm{-}S} \,l_1(r_m)} \right) \nonumber \\
		&\hspace*{2mm}\overset{\mathrm{(a)}}{>}  1 - \mathbb{E}\!\left[\!\left(1-\exp\!\left( \!\frac{-\eta T_1{I_{\textrm{T}}} }{p_{\textrm{t}} D_{T_m\textrm{-}S} \,l_1(r_m)}  \right) \! \right)^{\! N_{T\textrm{-}S}} \!\right] \nonumber \\
		&\hspace*{2mm}\overset{\mathrm{(b)}}{=} \sum_{n=1}^{N_{T\textrm{-}S}}(-1)^{n+1}\binom{N_{T\textrm{-}S}}{n}\mathbb{E}_{I_{\textrm{T}}}\! \left[\exp\!\left(\! \frac{-n \eta T_1{I_{\textrm{T}}} 16 \pi^2 f_1^2  r_m^{\alpha_1}} {p_{\textrm{t}} D_{T_m\textrm{-}S} c^2}\right) \!\right] \nonumber \\
		&\hspace*{2mm}\overset{\mathrm{(c)}}{=} \sum_{n=1}^{N_{T\textrm{-}S}}(-1)^{n+1}\binom{N_{T\textrm{-}S}}{n}\mathbb{E}_{I_{\textrm{T}}}\left[ \exp\left(-s{I_{\textrm{T}}} \right)\right] \nonumber\\
		&\hspace*{2.5mm}= \sum_{n=1}^{N_{T\textrm{-}S}}(-1)^{n+1}\binom{N_{T\textrm{-}S}}{n}\mathcal{L}_{I_{\textrm{T}}}(s),
	\end{align}
	where (a) is a tight upper bound when $N_{T\textrm{-}S}$ is small \cite{alzer1997some}, i.e., $\mathbb{P}\left(|h|^{2}\! <\! \psi\right)\! <\! \left(1\! -\! \exp(-\psi\eta)\right)^{N_{T\textrm{-}S}}$ with $\eta\! =\! N_{T\textrm{-}S}(N_{T\textrm{-}S}!)^{-\frac{1}{N_{T\textrm{-}S}}}$, (b) is obtained by binomial theorem and (c) is obtained by denoting $s\! =\! \frac{n \eta T_1 16 \pi^2 f_1^2  r_m^{\alpha_1}} {p_{\textrm{t}} D_{T_m\textrm{-}S} c^2} $. This completes the proof.
\end{proof}

\subsection{Proof of Lemma~\ref{L4}}\label{ApF}

\begin{proof}
	\begin{align}\label{LISproof} 
		&\! \mathcal{L}_{I_{\textrm{S}}}\!\left({s}'\right)\! =\!
		\mathbb{E}\!\! \left[\! \exp\!\left(\! -{s}'\!\!\!\!\!\!\! \sum_{S_n\in\Phi_L {\backslash} \{S_m\}}\!\!\!\!\! p_n D_{S_n\textrm{-}ES}\left|h_{S_n\textrm{-} ES }\right|^2 l_2(R_{S_n})\!\! \right)\!\! \right] \nonumber \\
		&=  \mathbb{E}\! \left[\! \prod_{S_n\in\Phi_L {\backslash} \{S_m\}}\!\! \!\!\!\! \exp\! \left(\! -{s}' p_n  D_{S_n\! -\! ES}\left|h_{S_n\textrm{-} ES }\right|^2 l_2(R_{S_n})\! \right)\! \right] \nonumber \\
		&=  \mathbb{E}_{N_I}\! \Bigg[\! \prod_{S_n\in\Phi_L {\backslash} \{S_m\}} \int_{H}^{r_{\textrm{max}}} \frac{2r_n'}{r_{\textrm{max}}^2-H^2} \nonumber \\
		&\times \! \mathbb{E}_{\left|h_{S_n\textrm{-} ES }\right|^2}\! \left[\! \exp\! \left(\!- {s}' p_n  D_{S_n\textrm{-}ES}\!\left|h_{S_n\textrm{-} ES }\right|^2 l_2(r_n')\! \right)\! \right]\! \mathrm{d}r_n'\! \Bigg]  \nonumber \\
		&\,{\overset{\mathrm{(a)}}=} \sum_{n_I=1}^{N_{\textrm{S}}-1}\! \binom{N_{\textrm{S}}}{n_I} P_I'^{n_I} (1-P_I')^{N_{\textrm{S}}-n_I-1}\! \int_{H}^{r_{\textrm{max}}}\! \frac{2r_n'	}{r_{\textrm{max}}^2-H^2} \nonumber \\
		&\times \mathbb{E}_{|h_{S_n\textrm{-} ES }|^2}\! \left[\exp\!\left(\!- \frac{{t_0}' \! \left|h_{S_n\textrm{-} ES }\right|^2 \!}{r_n'^{\alpha_2}} \right)\right]\!\!  \mathrm{d}r_n',
	\end{align}
	where (a) is obtained by the fact that $N_I$ follows the BPP and by denoting $t_0'\! =\! \frac{t \zeta(\beta\! -\! \delta) T_2 p_n D_{S_n\textrm{-} ES }\, r_m'^{\alpha}}{p_{\textrm{s}}}$.
	
	We now derive $\mathbb{E}_{{|h_{S_n\textrm{-} ES }|^2}}\left[\exp\left(\!- \frac{t_0' \left|h_{S_n\textrm{-} ES }\right|^2 \!}{r_n'^{\alpha_2}} \right)\right]$. The MGF of the SR fading model is defined as $M_S(x)\! =\! \mathbb{E}\left[\exp(-xS)\right]\! =\!\frac{(2\bar{c}q)^q(1+2\bar{c}x)^{q-1}}{((2\bar{c}q+\Omega )(1+2\bar{c}x)-\Omega )^q}$ \cite{2003A}. Thus, we have
	\begin{align} \label{Ehx} 
		& \mathbb{E}_{|h_{S_n\textrm{-} ES }|^2}\left[\exp\left(-\frac{{t_0}' \left|h_{S_n\textrm{-} ES }\right|^2 }{r_n'^{\alpha_2}} \right)\right] \nonumber \\
		&\hspace*{10mm}= \frac{(2\bar{c}q)^q(1+2\bar{c}{t_0}' r_n'^{-\alpha_2})^{q-1}}{\big((2\bar{c}q+\Omega )(1+2\bar{c}{t_0}'  r_n'^{-\alpha_2})-\Omega\big)^q} .
	\end{align}	 
	Next the probability of success $P_I'$ can be expressed as
	\begin{align}\label{P_I'} 
		P_I' =& \frac{\mathcal{S}(\mathcal{A}_{\textrm{S}})}{\mathcal{S}(\mathcal{A}_{\textrm{S}'})} = \frac{2\pi (r_{\textrm{e}} + H)^2(1 - \cos(\theta_3))}{4\pi(r_{\textrm{e}} + H)^2} \nonumber \\
		=& \frac{1-\cos(\theta_3)}{2},
	\end{align}
	where $\mathcal{A}_{\textrm{S}'}$ is the spherical surface at a height of $r_{\textrm{e}}+H$ at which satellites are located, and $\cos(\theta_3)\! =\! \frac{H + r_{\textrm{e}} - r_{\textrm{max}} \cos\big(\frac{\varphi_{\textrm{s}}}{2}\big)}{r_{\textrm{e}}}$.
	
	Substituting (\ref{Ehx}) and (\ref{P_I'}) into (\ref{LISproof}) leads to (\ref{eqL4}).
\end{proof}	

\subsection{Proof of Theorem~\ref{T2}}\label{ApE}

\begin{proof}
	Using the Kummer’s transform of the hypergeometric function \cite{5313912}, the PDF of \emph{$|h|^{2}$} can be rewritten as $f_{\left|h\right|^{2}}(x)\! =\! \sum\limits_{k=0}^{\infty}\Psi\left(k\right)x^{k}\exp\left(-\left(\beta-\delta\right)x\right)$, where $\Psi (k)\! =\! \frac{\left(-1\right)^{k}\kappa\delta^{k}}{\left(k!\right)^{2}}(1-q)_{k}$. Then the CDF of \emph{$|h|^{2}$} can be expressed as
	\begin{align}\label{eq5} 
		F_{|h|^{2}}\left(x\right) =& \sum_{k=0}^{\infty} \Psi (k) \int_{0}^{x} t^{k}\exp(-(\beta-\delta)t) dt \nonumber \\
		=& \sum_{k=0}^{\infty} \frac{\Psi\left(k\right)}{\left(\beta-\delta\right)^{k+1}}\gamma\left(k+1,\left(\beta-\delta\right)x\right).
	\end{align}
	Hence, we can obtain $P_{\mathrm{cov}}^{\textrm{S-ES}}$ as
	\begin{align}\label{A5-2} 
		& P_{\mathrm{cov}}^{\textrm{S-ES}} \triangleq  \mathbb{P}\left[ \frac{p_{\textrm{s}} \left|h_{S_m\textrm{-} ES}\right|^2 D_{S_m\textrm{-} ES} l_2(R_{S_m})} {I_{\textrm{S}}} \geq T_2  \right] \nonumber \\
		& = 1\! -\! \int_H^{r_{\textrm{max}}}\!\!\! \mathbb{P}\left[\! \left|h_{S_m\textrm{-} ES}\right|^2\! \le\! \frac{T_2 I_{\textrm{S}}}{p_{\textrm{s}} D_{S_m\textrm{-} ES}\, l_2(r_m')} | R_{S_m}\! =\! r_m'\! \right] \nonumber \\
		&\hspace*{3mm}\times	\frac{2r_m'}{r_{\textrm{max}}^2-H^2}\mathrm{d}r_m'.
	\end{align}
	Now we derive $\mathbb{P}\left[\left|h_{S_m\textrm{-}ES}\right|^2 \le \frac{T_2 I_{\textrm{S}}}{p_{\textrm{s}} D_{S_m\textrm{-} ES}l_2(r_m')} | R_{S_m}=r_m'\! \right]$:
	\begin{align} 
		& \mathbb{P}\left[ \left|h_{S_m\textrm{-}ES}\right|^2 \le \frac{T_2 I_{\textrm{S}}}{p_{\textrm{s}} D_{S_m\textrm{-} ES}\, l_2(r_m')} | R_{S_m}=r_m' \right] \nonumber \\
		&\hspace*{2mm} = \mathbb{E}\!\!\left[\! \sum_{k=0}^{\infty}\! \frac{\Psi\left(k\right)}{(\beta-\delta) ^{\,k+1}}  \gamma \!\!\left( k+1,(\beta - \delta) \frac{T_2 I_{\textrm{S}}}{p_{\textrm{s}} D_{S_m\textrm{-} ES}\, l_2(r_m')}\right)\! \right] \nonumber \\  
		&\hspace*{2mm} \overset{\mathrm{(a)}}{\approx} \mathbb{E} \left[ \sum_{k=0}^{\infty } \frac{\Psi\left(k\right)}{(\beta-\delta)^{k+1}} \Gamma (k\!+\!1) \right. \nonumber \\ 	
		&\hspace*{6mm} \times \left. \left(1-\exp\left(- \frac{\zeta(\beta-\delta) T_2 I_{\textrm{S}}}{p_{\textrm{s}} D_{S_m\textrm{-} ES} \,l_2(r_m')} \right)\right)^{k+1}\right] \nonumber \\
		&\hspace*{2mm} \overset{\mathrm{(b)}}{=} \sum_{k=0}^{\infty} \frac{\Psi\left(k\right)}{(\beta-\delta)^{k+1}} \Gamma (k+1) \sum_{t=0}^{k+1} \binom{k+1}{t} \nonumber \\
		&\hspace*{6mm}	\times(-1)^t \mathbb{E}\left[\exp(-{s}'I_{\textrm{S}})\right] \nonumber
			\end{align}	
				\begin{align}\label{A5-3} 
		&\hspace*{2.5mm} = \sum_{k=0}^{\infty} \frac{\Psi\left(k\right)}{(\beta-\delta)^{k+1}} \Gamma (k+1) \sum_{t=0}^{k+1} \binom{k+1}{t}(-1)^t 	\mathcal{L}_{I_{\rm{S}}}({s}'), 
	\end{align}	
	where (a) is approximated by using $\gamma(k\! +\! 1,x)\! <\! \Gamma (k\! +\! 1)(1\! -\! \exp(-\zeta x))^{k+1}$\cite{alzer1997some}, $\zeta\! =\! (\Gamma(k\! +\! 2))^{-\frac{1}{k+1}}$,  and (b) is obtained from binomial theorem with ${s}'\! =\! \frac{t \zeta(\beta\! -\! \delta) T_2}{p_{\textrm{s}} D_{S_m\textrm{-}ES} l_2(r_m')}$. 
	
	Substituting (\ref{A5-3})  into (\ref{A5-2}) leads to $P_{\mathrm{cov}}^{\textrm{S-ES}}$ (\ref{eqT2}).
\end{proof}

\subsection{Proof of Theorem~\ref{T3}}\label{ApG}

\begin{proof}
	From (\ref{eqAER}), Bayes rule and (\ref{eqS4-1-1}), we have
	\begin{align*} 
		& \bar{C}^{\rm{T}\textrm{-}\rm{S}} \triangleq \mathbb{E}\left[\log_2(1 + \mathrm{SINR_1})|R_{T_m}\leq r_{\textrm{max}}\right] \nonumber \\
		& \overset{\mathrm{(a)}}=\! \int_{H}^{r_{\textrm{max}}}\!\!\!\! \int_{t>0}\!\!\!\! \mathbb{E}\left[\mathbb{P}\!\!\ \Big(\log_2 (1+ \mathrm{SINR_1} >t\Big)\right]\! f_{R_{T_m}}(r_m)\,\mathrm{d}t\,\mathrm{d}r_{m} \nonumber \\
		&=\! \int_{H}^{r_{\textrm{max}}}\!\!\!\! \int_{t>0}\!\!\! \mathbb{E}\! \left[\mathbb{P}\left( |h_{T_m \textrm{-} S}|^2\! >\! \frac{I_{\textrm{T}}(2^t\! -\! 1)}{p_{\textrm{t}}D_{T_m\textrm{-}S}l_1(r_m)}\right)\right] \nonumber \\
		&\hspace*{4mm}\times \left(1-\frac{r_m^2-H^2}{4(r_{\textrm{e}}+H)r_{\textrm{e}} }\right)^{N_{\textrm{S}}-1}\!\!\!\! \frac{N_{\textrm{S}}\, r_m }{2(r_{\textrm{e}}+H)r_{\textrm{e}} } \mathrm{d}t \,\mathrm{d}r_{m} 
	\end{align*}
	\begin{align}\label{eqA7} 
		&\overset{\mathrm{(b)}}{\approx} \int_{H}^{r_{\textrm{max}}} \!\!\!\!\int_{t>0}\!\! \left(1-\mathbb{E}\!\left[\!\left(\!1\!-\!\exp\!\left( \frac{-\eta {I_{\textrm{T}}}(2^t\! -\! 1) } {p_{\textrm{t}} D_{T_m\textrm{-}S} \,l_1(r_m)}  \right)\! \right)^{\!\! N_{T\textrm{-}S}} \right] \right) \nonumber \\
		&\hspace*{4mm}\times \left(1-\frac{ r_m^2-H^2}{4(r_{\textrm{e}}+H)r_{\textrm{e}} }\right)^{N_{\textrm{S}}-1}\!\!\!\! \frac{N_{\textrm{S}}\, r_m }{2(r_{\textrm{e}}+H)r_{\textrm{e}} } \mathrm{d}t \,\mathrm{d}r_{m} \nonumber \\
		&\overset{\mathrm{(c)}}{=} \int_{H}^{r_{\textrm{max}}} \!\!\!\!\int_{t>0} \sum_{n=1}^{N_{T\textrm{-}S}}(-1)^{n+1}\binom{N_{T\textrm{-}S}}{n}\mathbb{E}_{I_{\textrm{T}}}\! \! \left[\exp\left(-s_1 {I_{\textrm{T}}}\right)\right] \nonumber \\
		&\hspace*{4mm}\times\! \left(\! 1-\frac{ r_m^2-H^2}{4(r_{\textrm{e}}+H)r_{\textrm{e}} }\right)^{N_{\textrm{S}}-1}\!\!\!\! \frac{N_{\textrm{S}}\, r_m}{2(r_{\textrm{e}}+H)r_{\textrm{e}}} \mathrm{d}t \,\mathrm{d}r_{m}\nonumber\\
		&\hspace*{0.5mm}= \int_{H}^{r_{\textrm{max}}} \!\!\!\!\int_{t>0} \sum_{n=1}^{N_{T\textrm{-}S}}(-1)^{n+1}\binom{N_{T\textrm{-}S}}{n}\mathcal{L}_{I_{\rm{T}}}(s_1) \nonumber \\
		&\hspace*{4mm}\times\! \left(\! 1-\frac{ r_m^2-H^2}{4(r_{\textrm{e}}+H)r_{\textrm{e}} }\right)^{N_{\textrm{S}}-1}\!\!\!\! \frac{N_{\textrm{S}}\, r_m}{2(r_{\textrm{e}}+H)r_{\textrm{e}}} \mathrm{d}t \,\mathrm{d}r_{m},
	\end{align}
	where (a)~follows form the fact that if the random variable $X$ involved is positive, $\mathbb{E}[X]\! =\! \int_{t>0}\mathbb{P}(X>t)\mathrm{d}t$, (b)~is obtained by the tight upper bound when $N_{T\textrm{-}S}$ is small \cite{alzer1997some}, that is, $\mathbb{P}\left[|h|^{2}\! <\! \psi\right]\! <\! \left(1\! -\! \exp(-\psi\eta)\right)^{N_{T\textrm{-}S}}$ with $\eta\! =\! N_{T\textrm{-}S}(N_{T\textrm{-}S}!)^{-\frac{1}{N_{T\textrm{-}S}}}$, and (c)~is obtained by binomial theorem and by denoting $s_1\! =\! \frac{n \eta (2^t\! -\! 1) 16 \pi^2 f_1^2  r_m^{\alpha_1}} {p_{\textrm{t}} D_{T_m\textrm{-}S} c^2} $.  This completes the proof.
\end{proof}	

\subsection{Proof of Theorem~\ref{T4}}\label{ApH}

\begin{proof}
	Starting from the definition (\ref{eqAER}), we have
	\begin{align} 
		& \bar{C}^{\rm{S}\textrm{-}\rm{ES}} \triangleq \mathbb{E}\left[\log_2(1 + \mathrm{SINR}_2)\right] \nonumber \\
		& \overset{\mathrm{(a)}}=\! \mathbb{E}\!\!\left[\int_{t>0}\!\! \mathbb{P}\Big(\!\log_2 (1\!+\!\frac{p_{\textrm{s}} D_{S_m\textrm{-}ES}\!\left|h_{S_m\textrm{-}ES}\right|^2 \!l_2(r_m')} {I_{\textrm{S}}+\sigma^2}\!>\!t\!\Big)\mathrm{d}t \!\right] \nonumber \\
		&= \mathbb{E}\! \left[\int_{t>0}\!\!\! \mathbb{P}\Big( \left|h_{S_m\textrm{-}ES}\right|^2\! >\! \frac{ (I_{\textrm{S}}+\sigma^2)(2^t\! -\! 1)}{p_{\textrm{s}} D_{S_m\textrm{-}ES} \,l_2(r_m')}\Big)  \mathrm{d}t \right]\! \!\nonumber \\
		\end{align}
		\begin{align}\label{eqA8} 
		&= \int_{t>0}\!\!\! \left(1\!-\!\int_H^{r_{\textrm{max}}}\!\!\!\!\! E\left[\mathbb{P}\!\!\left(\left|h_{S_m\textrm{-} ES}\right|^2\! \le\!\! \frac{(I_{\textrm{S}}+\sigma^2) (2^t\! -\! 1) }{p_{\textrm{s}} D_{S_m\textrm{-}ES}\, l_2(r_m')}  \!\! \right)\!  \right]\right. \nonumber \\
		&\hspace*{5mm}\times  \left. \frac{2r_m'}{r_{\textrm{max}}^2-H^2}\mathrm{d}r_m' \right)\mathrm{d}t \nonumber \\
		&= \int_{t>0}\! \Bigg(1\! - \!\int_H^{r_{\textrm{max}}}\!\!\!\! \mathbb{E} \Bigg[ \sum_{k=0}^{\infty} \frac{\Psi\left(k\right)}{(\beta-\delta)^{\,k+1}} \nonumber \\ 
		&\hspace*{5mm}\times\!\gamma\! \left(\! k\! + \!1,(\beta\! -\! \delta) \frac{(I_{\textrm{S}}+\sigma^2)(2^t\! -\! 1)}{p_{\textrm{s}} D_{S_m\textrm{-}ES} l_2(r_m')}\! \right)\! \Bigg]\! \frac{2r_m'}{r_{\textrm{max}}^2\! -\! H^2}\mathrm{d}r_m' \!\! \Bigg) \mathrm{d}t \nonumber\\
		& \overset{\mathrm{(b)}}{\approx} \int_{t>0}\! \Bigg(1\! -\! \int_H^{r_{\textrm{max}}}\!\!\!\! \mathbb{E}\!\Bigg[ \sum_{k=0}^{\infty} \frac{\Psi\left(k\right)}{(\beta-\delta)^{\,k+1}} \Gamma (k+1)\frac{2r_m'}{r_{\textrm{max}}^2-H^2} \nonumber \\ 	 
		&\hspace*{5mm} \times\!\! \left(\! 1\! -\! \exp\! \left(\! - \frac{\zeta(\beta\! -\! \delta) (2^t\! -\! 1) (I_{\textrm{S}}\! +\!\sigma^2)}{p_{\textrm{s}} D_{S_m\textrm{-}ES}\,l_2(r_m')} \right)\right)^{k+1}\Bigg]	\mathrm{d}r_m'\Bigg) \mathrm{d}t \nonumber \\ 
		&\overset{\mathrm{(c)}}= \int_{t>0}\! \Bigg(1\! -\! \int_H^{r_{\textrm{max}}}\!\!\! \mathbb{E}\!\Bigg[ \sum_{k=0}^{\infty} \frac{\Psi\left(k\right)}{(\beta-\delta)^{\,k+1}} \Gamma (k+1)\frac{ 2r_m'	}{r_{\textrm{max}}^2-H^2} \nonumber \\ 	 
		&\hspace*{5mm} \times\!\! \sum_{u=0}^{k+1}\! \binom{k+1}{u} (-1)^u \mathbb{E}_{I_{\textrm{S}}}\! \left[ \exp(-{s}_1'(I_{\textrm{S}}+\sigma^2))\! \right]\! \Bigg]\!	\mathrm{d}r_m'\! \Bigg)\!\mathrm{d}t,   
	\end{align}
	where (a)~follows form the fact that if the random variable $X$ involved is positive, $\mathbb{E}[X]\! =\! \int_{t>0}\mathbb{P}(X\! >\! t)\mathrm{d}t$, (b)~is approximated by using $\gamma(k+1,x)\! <\! \Gamma (k+1)(1-\exp(-\zeta x))^{k+1}$ with $\zeta\! =\! (\Gamma(k\! +\! 2))^{-\frac{1}{k+1}}$ \cite{alzer1997some}, and (c)~is obtained from binomial theorem with ${s}_1'\! =\! \frac{u \zeta(\beta-\delta) (2^t\! -\! 1)}{p_{\textrm{s}} D_{S_m\textrm{-}ES} \,l_2(r_m')}$. 
\end{proof}	

\small
\bibliographystyle{IEEEtran}

\begin{thebibliography}{10}
	\providecommand{\url}[1]{#1}
	\csname url@samestyle\endcsname
	\providecommand{\newblock}{\relax}
	\providecommand{\bibinfo}[2]{#2}
	\providecommand{\BIBentrySTDinterwordspacing}{\spaceskip=0pt\relax}
	\providecommand{\BIBentryALTinterwordstretchfactor}{4}
	\providecommand{\BIBentryALTinterwordspacing}{\spaceskip=\fontdimen2\font plus
		\BIBentryALTinterwordstretchfactor\fontdimen3\font minus
		\fontdimen4\font\relax}
	\providecommand{\BIBforeignlanguage}[2]{{%
			\expandafter\ifx\csname l@#1\endcsname\relax
			\typeout{** WARNING: IEEEtran.bst: No hyphenation pattern has been}%
			\typeout{** loaded for the language `#1'. Using the pattern for}%
			\typeout{** the default language instead.}%
			\else
			\language=\csname l@#1\endcsname
			\fi
			#2}}
	\providecommand{\BIBdecl}{\relax}
	\BIBdecl
	
\bibitem{7123563} 
A.~Al-Fuqaha, \emph{et al.}, ``Internet of Things: A survey on enabling technologies, protocols, and applications,'' \emph{IEEE Commun. Surveys Tuts.}, vol.~17, no.~4, pp.~2347--2376, 4th~Quart.~2015.



\bibitem{9042251} 
E.~Yaacoub and M.-S.~Alouini, ``A key 6G challenge and opportunity--connecting the base of the pyramid: A survey on rural connectivity,'' \emph{Proc. IEEE}, vol.~108, no.~4, pp.~533--582, Apr.~2020.

\bibitem{add0} 
Z. Lin, H. Niu, Z. Chu, Y. He, X. Zhong, P. Xiao and K. An, ``Self-powered absorptive reconfigurable intelligent surfaces for securing satellite-terrestrial integrated networks,'' China Communications, vol. 21, no. 9, pp. 276-291, Sep. 2024.

\bibitem{9442378} 
M.~Centenaro, \emph{et al.}, ``A survey on technologies, standards and open challenges in satellite IoT,'' \emph{IEEE Commun. Surveys Tuts.}, vol.~23, no.~3, pp.~1693--1720, 3rd Quart.~2021.

\bibitem{9040264} 
M.~Giordani, \emph{et al.}, ``Toward 6G networks: Use cases and technologies,'' \emph{IEEE Commun. Mag.}, vol.~58, no.~3, pp.~55--61, Mar.~2020.

\bibitem{9049651} 
M.~Giordani and M.~Zorzi, ``Satellite communication at millimeter waves: A key enabler of the 6G era,'' in \emph{Proc. ICNC 2020} (Big Island, HI, USA), Feb.~17-20, 2020, pp.~383--388.


\bibitem{1995Stochastic} 
S.~N.Chiu, D.~Stoyan, W.~S.~Kendall, and J.~Mecke, \emph{Stochastic Geometry and Its Applications} (3rd edition). John Wiley \& Sons, 2013.

\bibitem{2002Random} 
J.~Dall and M.~Christensen, ``Random geometric graphs,'' \emph{Phys. Rev. E}, vol.~66, article\,no.\,016121, pp.~1--9, Jul.~2002.

\bibitem{add1} 
Z. Lin, \emph{et al.}, ``Refracting RIS aided hybrid satellite-terrestrial relay networks: Joint beamforming design and optimization,'' \emph{IEEE Trans. Aerosp. Electron. Syst.}, vol. 58, no. 4, pp. 3717--3724, Aug. 2022.

\bibitem{add2} 
Z. Lin, \emph{et al.}, ``Supporting IoT with rate-splitting multiple access in satellite and aerial integrated networks,'' \emph{IEEE Internet  Things J.}, vol. 8, no. 14, pp. 11123--11134, Jun. 2021.






\bibitem{add3} 
Q. Li, \emph{et al.}, ``Holographic metasurface-based beamforming for multi-altitude LEO satellite networks,''  \emph{IEEE Trans. Wireless Commun.}, vol. 24, no. 4, pp. 3103--3116, Apr. 2025.

\bibitem{yastrebova2020theoretical} 
A.~Yastrebova, \emph{et a}l., ``Theoretical and simulation-based analysis of terrestrial interference to LEO satellite uplinks,'' in \emph{Proc. GLOBECOM 2020} (Taipei, Taiwan, China), Dec.~7-11, 2020, pp.~1--6.

\bibitem{9676997} 
B.~Manzoor, A.~Al-Hourani, and B.~A.~Homssi, ``Improving IoT-over-satellite connectivity using frame repetition technique,'' \emph{IEEE Wireless Commun. Lett.}, vol.~11, no.~4, pp.~736--740, Apr.~2022.

\bibitem{al2021modeling} 
B.~A.~Homssi and A.~Al-Hourani, ``Modeling uplink coverage performance in hybrid satellite-terrestrial networks,'' \emph{IEEE Commun. Lett.}, vol.~25, no.~10, pp.~3239--3243, Oct.~2021.

\bibitem{9838778} 
C.~C.~Chan, B.~A.~Homssi, and A.~Al-Hourani, ``A stochastic geometry approach for analyzing uplink performance for IoT-over-satellite,'' in \emph{Proc. ICC 2022} (Seoul, South Korea), May~16-20, 2022, pp.~2363--2368.

\bibitem{10463093} 
A.~Talgat, M.~A.~Kishk, and M.-S.~Alouini, ``Stochastic geometry-based uplink performance analysis of IoT over LEO satellite communication,'' \emph{IEEE Trans. Aerosp. Electron. Syst.}, vol. 60, no. 4, pp. 4198--4213, Mar.~2024.

\bibitem{10679173} 
W.-Y.~Dong, S.~Yang, P.~Zhang, and S.~Chen, ``Stochastic geometry based modeling and analysis of uplink cooperative satellite-aerial-terrestrial networks for nomadic communications with weak satellite coverage,'' \emph{IEEE J. Sel. Areas Commun.}, vol.~42, no.~12, pp.~3428--3444, Dec.~2024,

\bibitem{9257490} 
A.~Al-Hourani and I.~Guvenc, ``On modeling satellite-to-ground path-loss in urban environments,'' \emph{IEEE Commun. Lett.}, vol.~25, no.~3, pp.~696--700, Mar.~2021.

\bibitem{Fundamental} 
M. Afshang and H. S. Dhillon, ``Fundamentals of modeling finite wireless networks using binomial point process,'' \emph{IEEE Trans. Wireless Commun.}, vol.~16, no.~5, pp.~3355–-3370, May 2017.

\bibitem{2021Stochastic} 
A.~Talgat, M.~A.~Kishk, and M.~S.~Alouini, ``Stochastic geometry-based analysis of LEO satellite communication systems,'' \emph{IEEE Commun. Lett.}, vol.~25, no.~8, pp.~2458--2462, Aug.~2021.

\bibitem{2020An} 
A.~Al-Hourani, ``An analytic approach for modeling the coverage performance of dense satellite networks,'' \emph{IEEE Wireless Commun. Lett.}, vol.~10, no.~4, pp.~897--901, Apr.~2021.

\bibitem{okati2021modeling} 
N.~Okati and T.~Riihonen, ``Modeling and analysis of LEO mega-constellations as nonhomogeneous Poisson point processes,'' in \emph{Proc. VTC-Spring 2021} (Helsinki, Finland), Apr.~25-28, 2021, pp.~1--5.

\bibitem{7869087} 
M.~Sellathurai, S.~Vuppala, and T.~Ratnarajah, ``User selection for multi-beam satellite channels: A stochastic geometry perspective,'' in \emph{Proc. 50th Asilomar Conf. Signals, Syst. Comput.} (Pacific Grove, CA, USA), Nov.~6-9, 2016, pp.~487--491.

\bibitem{8068989} 
O.~Y.~Kolawole, S.~Vuppala, M.~Sellathurai, and T.~Ratnarajah, ``On the performance of cognitive satellite-terrestrial networks,'' \emph{IEEE Trans. Cogn. Commun. Netw}, vol.~3, no.~4, pp.~668--683, Dec.~2017.

\bibitem{9678973} 
D.-H.~Jung, J.-G.~Ryu, W.-J.~Byun, and J.~Choi, ``Performance analysis of satellite communication system under the shadowed-Rician fading: A stochastic geometry approach,'' \emph{IEEE Trans. Commun.}, vol.~70, no.~4, pp.~2707--2721, Apr.~2022.


\bibitem{song2022cooperative} 
Z.~Song, \emph{et~al.}, ``Cooperative satellite-aerial-terrestrial systems: A stochastic geometry model,'' \emph{IEEE Trans. Wireless Commun.}, vol.~22, no.~1, pp.~220--236, Jan.~2023.


\bibitem{7880676} 
A.~Alkhateeb, \emph{et al.}, ``Initial beam association in millimeter wave cellular systems: Analysis and design insights,'' \emph{IEEE Trans. Wireless Commun.}, vol.~16, no.~5, pp.~2807--2821, May 2017.

\bibitem{8335329} 
Y.~Zhu, G.~Zheng, and M.~Fitch, ``Secrecy rate analysis of UAV-enabled mmwave networks using Mat{\'e}rn hardcore point processes,'' \emph{IEEE J. Sel. Areas Commun.}, vol.~36, no.~7, pp.~1397--1409, Jul.~2018.

\bibitem{9200666} 
M.~T.~Dabiri, \emph{et al.}, ``3D channel characterization and performance analysis of UAV-assisted millimeter wave links,'' \emph{IEEE Trans. Wireless Commun.}, vol.~20, no.~1, pp.~110--125, Jan.~2021.

\bibitem{9520123} 
X.~Zhang, \emph{et al.}, ``Stochastic geometry-based analysis of cache-enabled hybrid satellite-aerial-terrestrial networks with non-orthogonal multiple access,'' \emph{IEEE Trans. Wireless Commun.}, vol.~21, no.~2,
pp.~1272--1287, Feb.~2022.

\bibitem{Park} 
J.~Park, J.~Choi, and N.~Lee, ``A tractable approach to coverage analysis in downlink satellite networks,'' \emph{IEEE Trans. Wireless Commun.}, vol.~22, no.~2, pp.~793--807, Feb.~2023.
	
\bibitem{6932503} 
T.~Bai and R.~W.~Heath, ``Coverage and rate analysis for millimeter-wave cellular networks,'' \emph{IEEE Trans. Wireless Commun.}, vol.~14, no.~2, pp.~1100--1114, Feb.~2015.

\bibitem{7468478} 
M.~Arti, ``Two-way satellite relaying with estimated channel gains,'' \emph{IEEE Trans. Commun.}, vol.~64, no.~7, pp.~2808--2820, Jul.~2016.

\bibitem{9107497} 
K.~Guo,  \emph{et al.}, ``Performance analysis of hybrid satellite-terrestrial cooperative networks with relay selection,'' \emph{IEEE Trans. Veh. Technol.}, vol.~69, no.~8, pp.~9053--9067, Aug.~2020.

\bibitem{2003A} 
A.~Abdi, W.~Lau, M.-S.~Alouini, and M.~Kaveh, ``A new simple model for land mobile satellite channels: First- and second-order statistics,'' \emph{IEEE Trans. Commun.}, vol.~2, no.~3, pp.~519--528, May 2003.

\bibitem{jung2018outage} 
D.-H.~Jung and D.-G.~Oh, ``Outage performance of shared-band on-board processing satellite communication system,'' in \emph{Proc. VTC-Fall 2018} (Chicago, IL, USA), Aug.~27-30, 2018, pp.~1--5.

\bibitem{zhang2019performance} 
X.~Zhang, \emph{et al.}, ``Performance analysis of NOMA-based cooperative spectrum sharing in hybrid satellite-terrestrial networks,'' \emph{IEEE Access}, vol.~7, pp.~172321--172329, Dec.~2019.

\bibitem{alzer1997some} 
H.~Alzer, ``On some inequalities for the incomplete gamma function,'' \emph{Math. Comput.}, vol.~66, no.~218, pp.~771--778, Apr.~1997.

\bibitem{5313912} 
A.~L.~Fructos, R.~R.~Boix, and F.~Mesa, ``Application of Kummer’s transformation to the efficient computation of the 3-D green's function with 1-D periodicity,'' \emph{IEEE Trans. Antennas Propag.}, vol.~58, no.~1, pp.~95--106, Jan.~2010.




\end{thebibliography}


\begin{IEEEbiography}
	[{\includegraphics[width=1in,height=1.25in,clip,keepaspectratio]{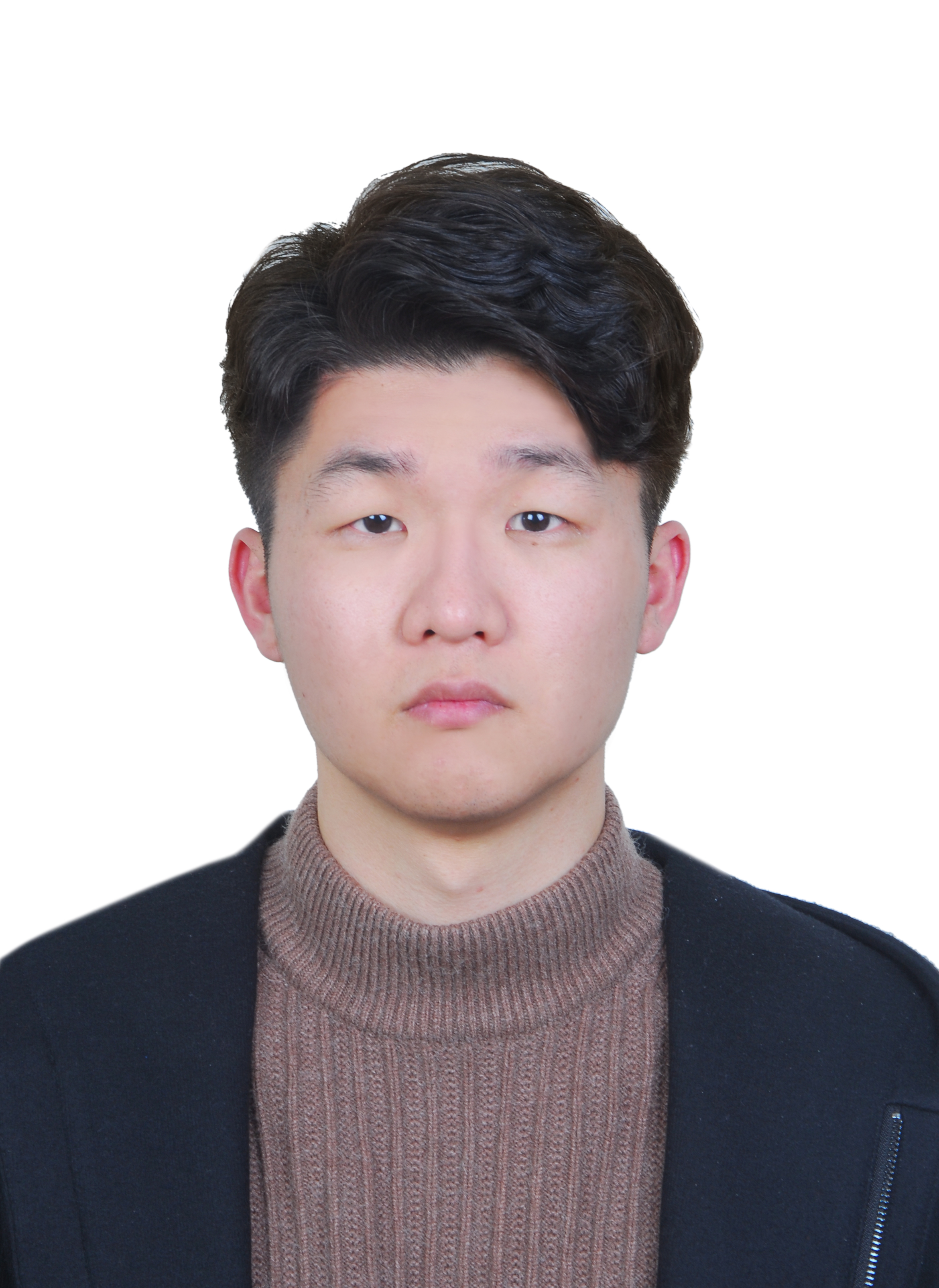}}]
	{Wen-Yu Dong}
	(Student Member, IEEE) 
	received the B.S. degree in electronic and information engineering from Sichuan University (SCU), China, in 2019. He is currently pursuing the Ph.D. degree in information and communication engineering with the School of Information and Communication Engineering, Beijing University of Posts and Telecommunications (BUPT), and the Key Laboratory of Universal Wireless Communications, Ministry of Education. His current research interests include highly dynamic mobile ad hoc network architecture and protocol stack design, routing design for mobile ad hoc networks, and integrated space-air-ground network modeling and performance analysis based on stochastic geometry.
\end{IEEEbiography}

\begin{IEEEbiography}
	[{\includegraphics[width=1in,height=1.25in,clip,keepaspectratio]{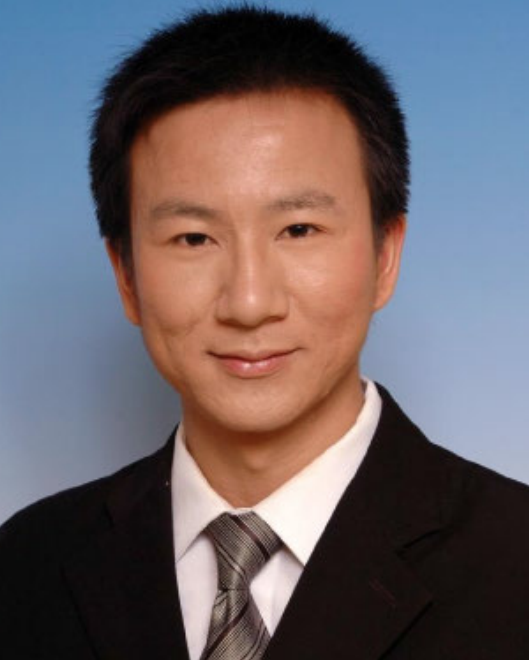}}]
	{Shaoshi Yang}
	(Senior Member, IEEE) 
	received the B.Eng. degree in information engineering from Beijing University of Posts and Telecommunications (BUPT), China, in 2006, and the Ph.D. degree in electronics and electrical engineering from the University of Southampton, U.K., in 2013. From 2008 to 2009, he was a Researcher with Intel Labs China. From 2013 to 2016, he was a Research Fellow with the School of Electronics and Computer Science, University of Southampton. From 2016 to 2018, he was a Principal Engineer with Huawei Technologies Co. Ltd., where he made significant contributions to the products, solutions, and standardization of 5G, wideband IoT, and cloud gaming/VR. He was a Guest Researcher with the Isaac Newton Institute for Mathematical Sciences, University of Cambridge. He is currently a Full Professor with BUPT. His research interests include 5G/5G-A/6G, massive MIMO, mobile ad hoc networks, distributed artificial intelligence, and cloud gaming/VR. He is a Standing Committee Member of the CCF Technical Committee on Distributed Computing and Systems. He received the Dean’s Award for Early Career Research Excellence from the University of Southampton in 2015, the Huawei President Award for Wireless Innovations in 2018, the IEEE TCGCC Best Journal Paper Award in 2019, the IEEE Communications Society Best Survey Paper Award in 2020, the Xiaomi Young Scholars Award in 2023, the CAI Invention and Entrepreneurship Award in 2023, the CIUR Industry-University-Research Cooperation and Innovation Award in 2023, and the First Prize of Beijing Municipal Science and Technology Advancement Award in 2023. He is an Editor of \emph{IEEE Transactions on Communications}, \emph{IEEE Transactions on Vehicular  Technology}, and \emph{Signal Processing} (Elsevier). He was also an Editor of \emph{IEEE Systems Journal} and \emph{IEEE Wireless Communications Letters}. For more details on his research progress, please refer to https://shaoshiyang.weebly.com/.
\end{IEEEbiography}

\begin{IEEEbiography}
	[{\includegraphics[width=1in,height=1.25in,clip,keepaspectratio]{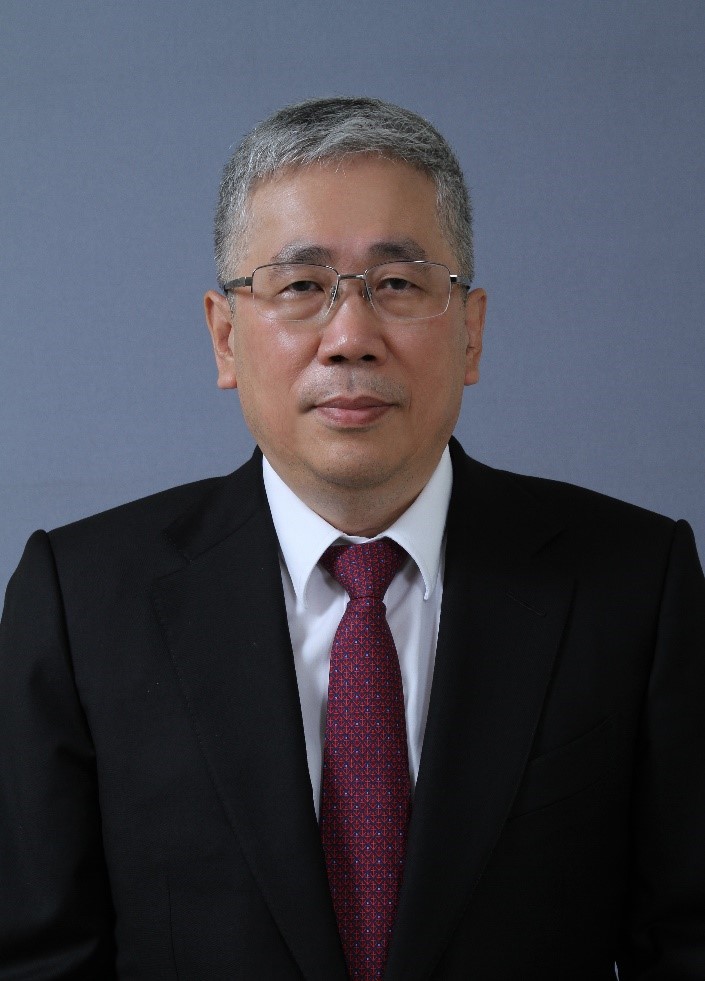}}]
	{Ping Zhang}
	(Fellow, IEEE)  is currently a professor of School of Information and Communication Engineering at Beijing University of Posts and Telecommunications, the director of State Key Laboratory of Networking and Switching Technology, a member of IMT-2020 (5G) Experts Panel, a member of Experts Panel for China’s 6G development. He served as Chief Scientist of National Basic Research Program (973 Program), an expert in Information Technology Division of National High-tech R$\&$D program (863 Program), and a member of Consultant Committee on International Cooperation of National Natural Science Foundation of China. His research interests mainly focus on wireless communication. He is an Academician of the Chinese Academy of Engineering (CAE).
\end{IEEEbiography}

\begin{IEEEbiography}[{\includegraphics[width=1in,height=1.25in,clip,keepaspectratio]{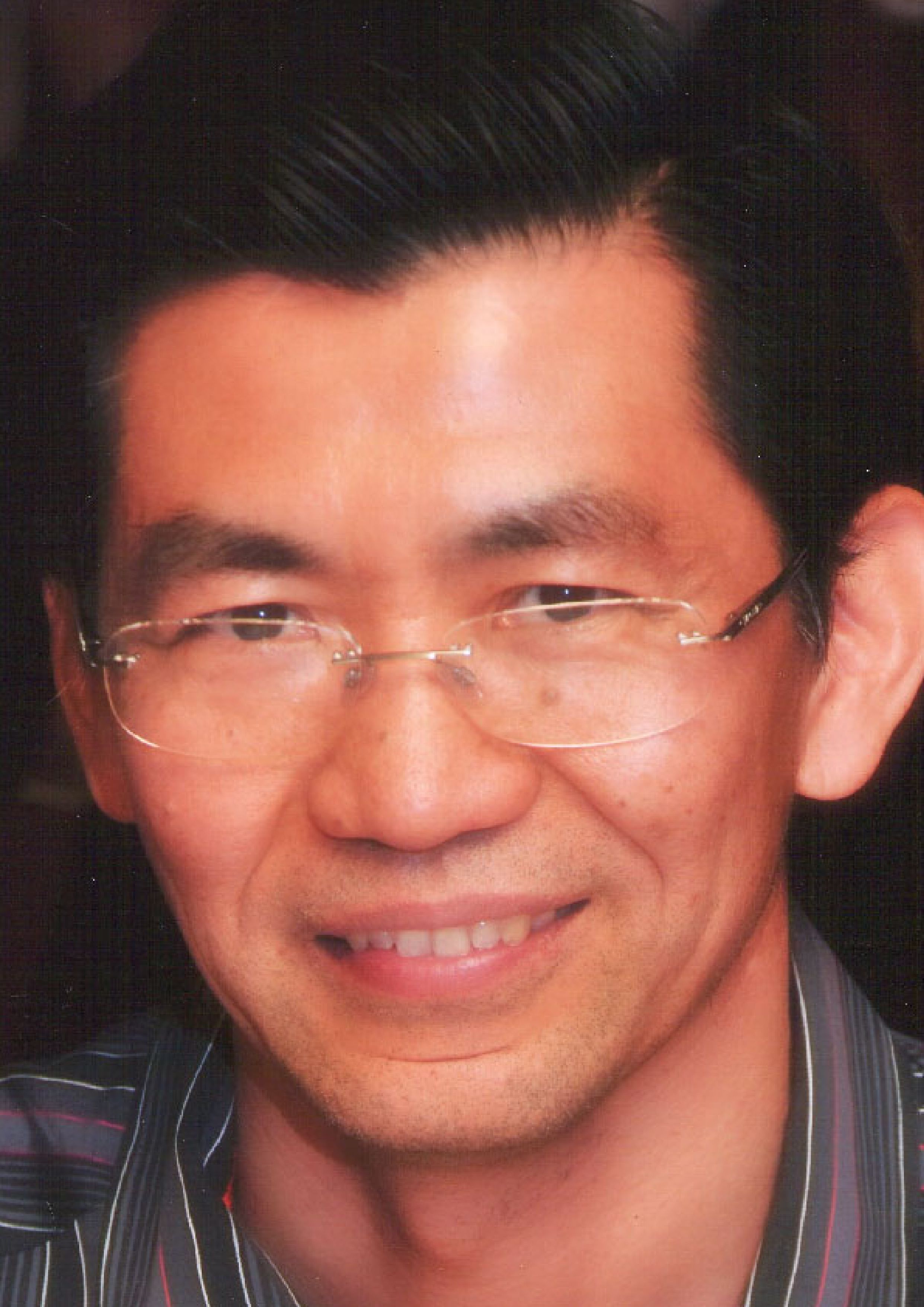}}]
	{Sheng Chen}
	(Life Fellow, IEEE) 
	received the B.Eng. degree in control engineering  from East China Petroleum Institute, Dongying, China, in 1982, the Ph.D. degree in control engineering from the City, University of London  in 1986, and the Doctor of Sciences (D.Sc.) degree from the University of Southampton, Southampton, U.K., in 2005. From 1986 to 1999, he held research and academic appointments at The University of Sheffield, U.K., The University of Edinburgh, U.K., and the University of Portsmouth, U.K. Since 1999, he has been with the School of Electronics and Computer Science, University of Southampton, where he is currently a Professor of intelligent systems and signal processing. He has published over 700 research articles. He has more than 20,000 Web of Science citations with an H-index of 63 and more than 40,000 Google Scholar citations with an H-index of 85. His research interests include adaptive signal processing, wireless communications, modeling and identification of nonlinear systems, neural network and machine learning, evolutionary computation methods, and optimization. He is a Fellow of the Royal Academy of Engineering, U.K., a Fellow of Asia–Pacific Artificial Intelligence Association, and a Fellow of IET. He was one of the original ISI Highly Cited Researchers in Engineering in March 2004.
\end{IEEEbiography}

\end{document}